\documentclass{article}
\newcommand{\BibTeX}{\rm B\kern-.05em{\sc i\kern-.025em b}\kern-.08em\TeX}

\usepackage{amsthm}
\usepackage{amsmath}
\usepackage{amssymb}
\usepackage{tikz-cd}

\usepackage{amsthm}
\usepackage{amsmath}
\usepackage{amssymb}
\usepackage{dsfont}
\usepackage{url}  
\usepackage[ruled,vlined]{algorithm2e}

\DeclareMathOperator*{\argmax}{arg\,max}
\DeclareMathOperator*{\argmin}{arg\,min}

\usepackage{mathrsfs}

\usepackage{graphicx}
\usepackage{subfig}

\usepackage{xspace}
\usepackage{bbm}
\usepackage{xcolor}
\usepackage{natbib}
\setcitestyle{numbers,open={[},close={]}} %

\usepackage{enumitem}

\newtheorem{theorem}{Theorem}

\newtheorem{definition}[theorem]{Definition}
\newtheorem{corollary}[theorem]{Corollary}
\newtheorem{proposition}[theorem]{Proposition}
\newtheorem{property}[theorem]{Property}%
\newtheorem{remark}[theorem]{Remark}

\date{}
\title{Learning Equilibria in Mean-Field Games:\\ Introducing Mean-Field PSRO}

\author{
\small{Paul Muller\thanks{Deepmind, corresponding author: pmuller@deepmind.com}} \and \small{Mark Rowland\footnotemark[1]} \and \small{Romuald Elie\footnotemark[1]} \and \small{Georgios Piliouras\thanks{Singapore University of Technology and Design}} \and \small{Julien Perolat\footnotemark[1]} \and \small{Mathieu Lauriere\thanks{Google Brain}} \and \small{Raphael Marinier\footnotemark[3]} \and \small{Olivier Pietquin\footnotemark[3]} \and \small{Karl Tuyls\footnotemark[1]}}

\begin{document}

\maketitle 

\begin{abstract}
    Recent advances in multiagent learning have seen the introduction of a family of algorithms that revolve around the population-based training method PSRO, showing convergence to Nash, correlated and coarse correlated equilibria. 
    Notably, when the number of agents increases, learning best-responses becomes exponentially more difficult, and as such hampers PSRO training methods. 
    The paradigm of mean-field games provides an asymptotic solution to this problem when the considered games are anonymous-symmetric.
    Unfortunately, the mean-field approximation introduces non-linearities which prevent a straightforward adaptation of PSRO.
    Building upon optimization and adversarial regret minimization, this paper sidesteps this issue and introduces mean-field PSRO, an adaptation of PSRO which learns Nash, coarse correlated and correlated equilibria in mean-field games. 
    The key is to replace the exact distribution computation step by newly-defined mean-field no-adversarial-regret learners, or by black-box optimization. 
    We compare the asymptotic complexity of the approach to standard PSRO, greatly improve empirical bandit convergence speed by compressing temporal mixture weights, and ensure it is theoretically robust to payoff noise.
    Finally, we illustrate the speed and accuracy of mean-field PSRO on several mean-field games, demonstrating convergence to strong and weak equilibria.
\end{abstract}

\section{Introduction}

This paper introduces a new mean-field reinforcement learning algorithm, Mean-Field Policy Space Response Oracles (MF-PSRO), guaranteed to converge to Nash, correlated and coarse-correlated equilibria in a large variety of games, without any hypothesis thereupon. 
Policy Space Response Oracles (PSRO) \citep{lanctot2017unified} is originally a two-player zero-sum game algorithm meant to be a generalization of double-oracle \citep{mcmahan2003planning}, fictitious play \citep{brown1951iterative}, and independent reinforcement learning \citep{matignon2012independent}. The algorithm's main loop is composed of two steps: given a policy set, compute an optimal distribution of play. Then, compute a best-response to this distribution, add it to the set and re-iterate.
Remarkably, recent years have shown the algorithm's versatility by demonstrating great advances in learning $N$-player equilibria using PSRO-derived approaches, managing to converge towards $\alpha$-Rank~\cite{omidshafiei2019alpharank,NEURIPS2019_510f2318} optimal strategy cycles~\cite{muller2020generalized}, or towards (coarse) correlated equilibria\footnote{A broad relaxation of Nash, correlated equilibria that is closely connected to regret minimization. It is sometimes referred to as Hannan consistency~\cite{nisan2007algorithmic,hart2013simple,monnot2017limits}.}~\cite{marris2021multiagent}. However, both latter approaches' convergence results rely on potentially fully exploring the space of deterministic strategies, which grows exponentially in the number of players. 
Computing a best response in the general case of randomized opponent strategies also becomes exponentially more complex as the number of players increases, even with symmetric simplifications such as anonymity~\cite{verma2020entropy}, centralized settings~\cite{khan2018scalable}, or fully cooperative settings~\cite{oroojlooyjadid2021review}. Although anonymity can allow polynomial-time approximation schemes for computing approximate Nash equilibria~\cite{daskalakis2007computing,daskalakis2008discretized}, in practice such algorithms are typically too slow for real life applications. A more promising way to address such complexity issues is by approximation in the case of symmetric games by considering asymptotic versions thereof, where the number of players is infinite and only their distribution matters: mean-field games~\cite{huang2006large,lasry2007mfg}.

The question of learning Nash equilibria in mean-field games has been receiving a growing amount of attention, and many methods have been recently proposed. Among these, we can distinguish those relying on fixed-point contraction~\cite{anahtarci2020q,guo2020general,xie2020provable}, fictitious-play~\cite{cardaliaguet2015learning, hadikhanloo2018finite, elie2020convergence,perrin2020fictitious} or online mirror descent~\cite{perolat2021scaling}.
Comparatively, learning correlated and coarse correlated equilibria in mean-field games has not yet, to the best of our knowledge, been studied. The literature has only started introducing  notions of mean-field correlated equilibria~\cite{campi2021correlated,deglinnocenti2018correlated}.

Our central question is: \textit{What are the modifications required for PSRO to successfully converge towards Nash, correlated and coarse correlated equilibria in mean-field games?}

In order to answer it, after introducing the framework of interest (Section \ref{sec:definitions}), we expand on the obstacles encountered when attempting to adapt PSRO to mean-field games (Section \ref{section:challenges_scaling}), identify and treat the cases where a straightforward adaptation is possible, then generalize the setting to mean-field games with finite states, actions and times. Note that the general treatment is fundamentally different for Mean-Field Nash equilibria (Section \ref{section:nash_convergence}), and for Mean-Field (coarse) correlated equilibria (Section \ref{section:cce_convergence}). Finally, we test our algorithms on a number of OpenSpiel~\cite{lanctot2020openspiel} games in Section \ref{section:experiments}, demonstrating convergence, and, where possible, comparing with alternative benchmarks. %

\section{Background}\label{sec:definitions}

\subsection{Definitions}

Given a set $\mathcal{Y}$, we name $\Delta(\mathcal{Y})$ the set of distributions over $\mathcal{Y}$. 

A game is a tuple $(\mathcal{X}, \mathcal{A}, r, P, \mu_0)$ where $\mathcal{X}$ is the finite set of states, $\mathcal{A}$ is the finite set of actions, $r: \mathcal{X} \times \mathcal{A} \times \Delta(\mathcal{X}) \rightarrow \mathbb{R}$ is the reward function where $\Delta(\mathcal{X})$ is the set of probability distributions over $\mathcal{X}$, $P: \mathcal{X} \times \mathcal{A} \rightarrow \mathcal{X}$ is the state transition function, assumed not to depend over $\Delta(\mathcal{X})$, $\mu_0 \in \Delta(\mathcal{X})$ is the initial state occupancy measure. 

We take $\Pi$ the set of deterministic policies, which is \emph{finite} and whose convex hull spans all game policies, a policy being a function $\pi: \mathcal{X} \rightarrow \Delta(\mathcal{A})$. 

We name $J$ the expected payoff function $$J(\pi, \mu) := \sum_{x \in \mathcal{X}, a \in \mathcal{A}} \mu^\pi(x) \pi(x, a) r(x, a, \mu)$$ where $\mu^\pi$ is the expected state occupancy measure of a representative player playing policy $\pi$. We note that in Mean-Field games, a single player has no influence on the reward function, since it has no influence on the state distribution $\mu$: only behavior changes that are wide enough to modify the player state distribution $\mu$ can change the MDP's reward function. State occupancy measures can be defined in several ways, and our derivations apply to all: 
\begin{itemize}[leftmargin=0.35cm]
    \item $\gamma$-discounted: $\mu^\pi(x)\! =\! \mu_0(x) + \gamma \sum\limits_{x' \in \mathcal{X}}\sum\limits_{a \in \mathcal{A}}\! p(x | x', a) \pi(x', a) \mu^\pi(x')$
    \item Finite-horizon: $\mu^\pi_{s+1}(x) = \sum\limits_{x' \in \mathcal{X}}\sum\limits_{a \in \mathcal{A}} p(x | x', a) \pi(x', a) \mu_s^\pi(x')\;\;\;$ with $\mu^\pi_0 = \mu_0$ (in which case another summation term over $s$ appears in $J$, or we assume states to contain current time $s$). 
\end{itemize}

Given policies $\pi_1, ..., \pi_n \in \Pi$, we call \textbf{restricted game} the stateless game where players choose one policy among $\{\pi_i | 1 \leq i \leq n \}$ at the beginning of the game, then keep playing it until the end. 

We also define \textbf{meta-games}, which are normal-form games whose payoff matrix for player 1 is, at row $i$ and column $j$, $J(\pi_i, \mu^{\pi_j})$ - and the transpose thereof for player 2. The complex relationship between these notions, which are equivalent in $N$-player games, is explored in Section \ref{section:challenges_scaling}. 

A \textbf{correlation device}, a notion introduced in~\cite{muller2022learning}, $\rho$ is a distribution over distributions of policies: $\rho \in \Delta(\Delta(\Pi))$, where $\Delta(\Pi)$ is the set of distribution over $\Pi$. It is used to sample population distributions $\nu \in \Delta(\Pi)$, from which individual population recommendations $\pi$ are in turn sampled: the distribution of policies over the whole mean-field population follows $\nu$ with probability $\rho(\nu)$. Given a sequence of distributions $(\nu_t)_t$ and a distribution $(\rho_t)_t$ over them, we write $(\rho_t, \nu_t)_t$ the correlation device recommending $\nu_t$ with probability $\rho_t$.

The \textbf{empirical play} of a sequence  $\nu_1, ..., \nu_T$ is the correlation device which uniformly selects one of the joint members of the sequence: $\forall 1 \leq t \leq T, \; \rho(\nu_t) = \frac{1}{T}$.

We write $\mu(\nu)$ the state occupancy measure of the population when policies are distributed according to $\nu$. In our case, the dynamics $p$ do not depend on $\mu$, so we have that $\mu(\nu) = \sum_\pi \nu(\pi) \mu^\pi$. We note that we restrict ourselves to the fully discrete setting. We also write $\pi(\nu)$ the policy resulting from sampling an initial policy according to $\nu$ and playing it until the end of the game. 

\subsection{Mean-field equilibria}

\begin{definition}[mean-field Nash equilibrium] 
    A \textbf{mean-field Nash equilibrium} (\textbf{MFNE}) is a policy $\pi$ such that, when the whole population plays $\pi$, no agent has an incentive to deviate, ie. $$ J(\pi', \mu^\pi) - J(\pi, \mu^\pi) \leq 0\,, \;\, \forall \pi' \in \Pi\,. $$
\end{definition}

The following two equilibria have been introduced by~\citet{muller2022learning}, and we refer the reader to this paper for more details on, and justifications of, their formulations. 

\begin{definition}[mean-field coarse-correlated equilibrium]
    A \textbf{mean-field coarse-correlated equilibrium} (\textbf{MFCCE}) is a correlation device $\rho$ from which players do not have an incentive to deviate before being given their recommendations, ie. $$ \mathbb{E}_{\nu \sim \rho, \pi \sim \nu} \left[ J(\pi', \mu(\nu)) - J(\pi, \mu(\nu)) \right] \leq 0\,, \;\, \forall \pi' \in \Pi. $$
\end{definition}

\begin{definition}[mean-field correlated equilibrium]
    A \textbf{mean-field correlated equilibrium} (\textbf{MFCE}) is a correlation device $\rho$ such that players do not have an incentive to deviate even after being given their recommendations, ie. $$ \rho(\pi) \mathbb{E}_{\nu \sim \rho(\cdot | \pi)} \left[ J(\pi', \mu(\nu)) - J(\pi, \mu(\nu)) \right] \leq 0\,, \;\, \forall \pi, \pi'\in\Pi $$
\end{definition}
where $\rho(\pi) = \sum\limits_{\nu} \nu(\pi) \rho(\nu)$ assuming $\rho$ only has atoms, and $\rho(\nu | \pi) = \frac{\nu(\pi) \rho(\nu)}{\sum\limits_{\nu'} \nu'(\pi) \rho(\nu')}$.

As usual, $\epsilon$-variants of these equilibria are defined by changing $0$ in the r.h.s. of the above inequalities by $\epsilon > 0$: these are approximate equilibria where one may only gain up to $\epsilon$ by deviating unilaterally (For Coarse-Correlated Equilibria) or per-policy (For Correlated Equilibria).

When applied to restricted games, we call these equilibria, restricted MFNE, restricted MFCE and restricted MFCCE.

\subsection{PSRO in $N$-player games}\label{sec:standard-psro}

PSRO~\cite{lanctot2017unified} is a generalization of Double Oracle~\cite{mcmahandoubleoracle}, and as such is an iterated best-response algorithm for computing Nash equilibria in $N$-player games. The algorithm, presented in Algorithm \ref{alg:psro} initiates  with sets containing random policies. At each iteration, an optimal policy distribution is computed over the policy sets, and a best response to this distribution is computed for each player. If all best responses were already in each player's policy set, the algorithm terminates; it continues otherwise.

\begin{algorithm}%
\SetAlgoLined
\KwResult{Policy sets $(\Pi_k^* = \{\pi_1^k, ..., \pi_n^k\})_{k=1..K}$ for all $K$ players, policy distributions $(\nu_k^*)_{k=1..N}$}
 $\forall k, \; \Pi_k^1 = \{ \pi_k^1 \}$ with $\pi_k^1$ any policy, $\nu_k(\pi_k^1) = 1.0$, $n = 1$\;
 \While{$(\Pi_{n+1} \setminus \Pi_n) \neq \emptyset$}{
  $\forall k, \; \Pi_k^{n+1} = \Pi_k^{n} \cup \{ BR_k(\nu) \}$ \;
  $n = n + 1$\;
  Fill payoff tensors $(T_k)_{k=1..K}$: $\forall x_1, ..., x_K, T_k(x_1, ..., x_K) = \text{Payoff}_k(\pi_{x_1}, ..., \pi_{x_K})$\;
  $\nu = \text{Meta-Solver}((T_k)_{k=1..K})$
 }
 \caption{PSRO(Meta-Solver) ($N$-player games)}
 \label{alg:psro}
\end{algorithm}

The original PSRO paper introduced several different meta-solvers (Uniform, Exact Nash and PRD, an approximate Nash solver), all of which were proven to make PSRO converge to a Nash equilibrium in two-player zero-sum games. Recent work has extended convergence to Alpharank~\cite{omidshafiei2019alpharank}-optimal subsets~\cite{muller2020generalized} and to correlated and coarse correlated equilibria~\cite{marris2021multiagent} in $N$-player games when using the right meta-solvers and best-responders. Crucially, the game specified by the payoff tensors that the meta-solver computes an equilibrium form is a normal-form matrix game. This yields a `linearity of evaluation' property; specifically, the payoffs when players make use of mixed strategies are straightforwardly computed from the payoff tensors specifying the payoffs of the pure strategies in the game.

In the rest of this paper, unless otherwise directly specified, we consider $n$ to be the current PSRO iteration.

\section{Challenges in scaling to mean-field games}\label{section:challenges_scaling}

Our central proposal in this paper is a generalisation of PSRO to the mean-field setting. We introduce two distinct algorithms for the computation of either MFNE or MFCE/MFCCE. Both MF-PSRO algorithms are described as   %
Algorithms~\ref{alg:mf-psro-nash} and \ref{alg:mf-psro-c-ce} below.

\begin{algorithm}%
\SetAlgoLined
\KwResult{Policy set  $\Pi^* = \{\pi_1, ..., \pi_n$\}, Policy Distribution $\nu^* \in \Delta(\Pi^*)$ yielding game Nash $\pi(\nu^*)$}
 $\Pi_1 = \{ \pi_1 \}$ with $\pi_1$ any policy, $\nu_1(\pi_1) = 1.0$, $n = 1$\;
 \While{$(\Pi_{n+1} \setminus \Pi_n) \neq \emptyset$}{
  $\Pi_{n+1} = \Pi_{n} \cup \{ BR(\mu^{\pi(\nu_n)}) \}$ \;
 
  $n = n + 1$ \;
  
  $\displaystyle \nu_n = \argmin_{\nu \in \Delta(\Pi_n)} \max_{i= 1,\ldots,n} J(\pi_i, \mu(\nu)) - J(\pi(\nu), \mu(\nu))$ \;
 }
 \caption{MF-PSRO(Nash)}\label{alg:mf-psro-nash}
\end{algorithm}

\begin{algorithm}%
\SetAlgoLined
\KwResult{Policy set  $\Pi^* = \{\pi_1, ..., \pi_n$\}, $\epsilon$-mean-field correlated equilibrium $\rho^* \in \Delta(\Delta(\Pi^*))$}
 $\Pi_0 = \emptyset$, $\Pi_1 = \{ \pi_1 \}$ with $\pi_1$ any policy, $\rho(\delta_{\pi_1}) = 1.0$, $n = 1$\;
 \While{$(\Pi_{n+1} \setminus \Pi_n) \neq \emptyset$}{
  
  (If CE) $\Pi_{n+1} = \Pi_{n} \cup \{ BR_{CE}(\pi_i, \rho_n) \; | \; \pi_i, \; \rho_n(\pi_i) > 0 \}$  \;
  (If CCE) $\Pi_{n+1} = \Pi_{n} \cup BR_{CCE}(\rho_n)$\;
  
  $n = n + 1$\;
  
  (If CE) $\displaystyle \rho_n = \argmin_{\rho \in \Delta(\Delta(\Pi_n))} \mathbb{E}_{\nu \sim \rho, \pi \sim \nu}[\max_{i=1..n} J(\pi_i, \mu(\nu)) - J(\pi, \mu(\nu))]$ \;
  \vspace{0.05cm}
  (If CCE) $\displaystyle \rho_n = \argmin_{\rho \in \Delta(\Delta(\Pi_n))} \max_{i=1,\ldots,n} \mathbb{E}_{\nu \sim \rho, \pi \sim \nu}[ J(\pi_i, \mu(\nu)) - J(\pi, \mu(\nu))]$\;
 }
 \caption{MF-PSRO((C)CE)}\label{alg:mf-psro-c-ce}
\end{algorithm}

These two algorithms have a very similar structure to the PSRO as described for $N$-player  games in Section~\ref{sec:standard-psro}; within the inner loop, a distribution is computed for the restricted game under consideration (either a Nash equilibrium, or a (coarse) correlated equilibrium), and new policies are derived as certain types of best response against the computed equilibrium. Keeping the same insight as~\cite{marris2021multiagent}, we define two different Best Responder functions $BR_{CE}$ and $BR_{CCE}$, for use with MF-PSRO in computing CEs and CCEs, respectively:
\begin{itemize}
    \item $BR_{CCE}(\rho) := \argmax\limits_{\pi^* \in \Pi} \sum_\nu \rho(\nu) J(\pi^*, \mu(\nu))$;
    \item $BR_{CE}(\pi_k, \rho) := \argmax\limits_{\pi^* \in \Pi} \sum_\nu \rho(\nu | \pi_k)  J(\pi^*, \mu(\nu))$.
\end{itemize}
We note that $BR_{CCE}(\rho)$ is the Best Response corresponding to a unilateral deviation from $\rho$, ie. deviating before having been given a recommendation, whereas $BR_{CE}(\pi_k, \rho)$ is the best response generated by deviating from recommendation $\pi_k$. 

Given these proto-algorithms, several important questions are immediately raised. First, are these algorithms guaranteed to return instances of the equilibria they seek to find? This is a purely mathematical question. Second, how should the restricted game equilibria in the inner loop be computed? As described in Section~\ref{sec:standard-psro}, the restricted game in usual applications of PSRO satisfies a `linearity of evaluation' property. Unfortunately, however, this linearity property is lost in the case of mean-field games, in which the representative player's payoff is generally non-linear as a function of the population occupancy measure. This lack of linearity presents a serious barrier in directly applying PSRO to mean-field games, and an important contribution of this paper is how to circumvent this barrier. We do however note that for a limited class of mean-field games, linearity is preserved; we describe the details of this case in Appendix \ref{appendix:linear_special_case}.

The next two sections treat the theoretical and implementation questions raised above for Nash equilibria, and for (coarse) correlated equilibria, in turn.

\section{Convergence to Nash equilibria}\label{section:nash_convergence}

\subsection{Existence and computation of restricted game equilibria}

In the inner loop of MF-PSRO(Nash), an important subroutine is the computation of a mean-field Nash equilibrium for the restricted game; namely, a distribution $\nu \in \Delta(\Pi_n)$ such that $$J(\pi', \mu(\nu)) - J(\pi(\nu), \mu(\nu)) \leq 0\,, \;\, \forall \pi' \in \{ \pi_1, ..., \pi_n \}.$$ We note that if at least one such $\nu$ exists, then the following optimization problem in the inner loop of MF-PSRO(Nash), which minimizes exploitability, will return a Nash equilibrium 
\begin{equation}\label{equation:nash}
    \nu^* = \argmin_{\nu \in \Delta_n} \max_{i= 1...n} J(\pi_i, \mu(\nu)) - J(\pi(\nu), \mu(\nu))    \, .
\end{equation}

Fortunately, the conditions of existence for a Nash equilibrium of the restricted game - so called restricted Nash equilibrium - only require continuity of $r$ with respect to $\mu$, as shown in the following theorem.

\begin{theorem}[Existence of restricted Nash equilibria]\label{theorem:restricted_nash_existence}
    If the reward function of the game is continuous with respect to $\mu$, then there always exists a restricted game Nash equilibrium.
\end{theorem}
\begin{proof}
    Let $\phi:\Delta(\Pi_n)\rightarrow 2^{\Delta(\Pi_n)}$ be the best-response map in the restricted game characterized by policies in the set $\Pi_n$: $$\forall \nu \in \Delta(\Pi_n), \quad \phi(\nu) := \argmax\limits_{\nu' \in \Delta(\Pi_n)} J(\pi(\nu'), \mu(\nu)).$$ %
    
     $\Delta(\Pi_n)$ is non-empty and convex, together with closed and bounded in a finite-dimensional space, and therefore compact.
     
    For all $\nu\in\Delta(\Pi_n)$, $\argmax\limits_{\nu' \in \Delta(\Pi_n)} J(\pi(\nu'), \mu(\nu)) \subseteq \Delta(\Pi_n)$ because $\Delta(\Pi_n)$ is closed, and $\phi(\nu)$ is therefore non-empty. 
    
    Let $\nu_1, \; \nu_2 \in \phi(\nu)$, $t \in [0, 1]$. 
    $$J(\pi(t \nu_1 + (1-t) \nu_2), \mu(\nu)) = t J(\pi(\nu_1), \mu(\nu)) + (1-t) J(\pi(\nu_2), \mu(\nu))$$ so $t \nu_1 + (1-t) \nu_2 \in \phi(\nu)$ and $\phi(\nu)$ is therefore convex.
    
    The proof of $\text{Graph}(\phi)$ being closed is provided in Appendix \ref{appendix:proof_graph_closedness}. It relies on the fact that since $r$ is continuous in $\mu$, so is $J$, and since the function $\nu \rightarrow J(\pi(\nu), \mu)$ is linear for all $\nu \in \Delta(\Pi_n)$, the function $(\nu_1, \nu_2) \rightarrow J(\pi(\nu_1), \mu(\nu_2))$ is bicontinuous, which is enough to ensure Graph closedness. We have all the hypotheses required to apply Kakutani's fixed point theorem~\cite{kakutani1941generalization}: there thus exists $\nu^* \in \Delta(\Pi_n) \text{ such that } \nu^* \in \phi(\nu^*)$, ie. $\nu^* = \argmax_{\nu'} J(\pi(\nu'), \mu(\nu^*))$, which means that $\forall \nu' \in \Delta(\Pi_n), J(\pi(\nu'), \mu(\nu^*)) \leq J(\pi(\nu^*), \mu(\nu^*))$, in other words: $\nu^*$ is a Nash equilibrium of the restricted game.
\end{proof}

Having established the existence of Nash equilibria for the restricted mean-field game in the inner loop of MF-PSRO(Nash), we now turn to the problem of how such an equilibrium can be (approximately) computed. As remarked earlier, due to the non-linearity of the restricted game, this problem is a non-linear (and potentially non-convex) optimisation problem over $\Delta(\Pi_n)$. Thus, the optimal solution of Equation~\eqref{equation:nash} can be, in the absence of any additional assumptions on the game, found via Black-Box optimization approaches, such as random search~ \cite{solis1981minimization}, Bayesian optimization~\cite{frazier2018tutorial}, evolutionary search (our experiments use CMA-ES~\cite{hansen2016cma}), or any other appropriate method for the considered game. 

\subsection{Convergence to Nash}

The termination condition of PSRO is the following: if at step $N+1$, the new policy $\pi_{n+1}$ produced by the algorithm is in $\Pi_n$, then the algorithm terminates. Given that each $\pi_i$ is a deterministic policy, and that the set of deterministic policies is finite, PSRO will therefore necessarily terminate.
We must only prove one thing:

\begin{proposition}[Termination-optimality] \label{proposition:nash_termination_optimality}
    If MF-PSRO(Nash) terminates, 
    it stops at a Nash equilibrium of the true game.
\end{proposition}
\begin{proof}
    If MF-PSRO(Nash) terminates at step $n$, then
    \begin{align*}
        \pi^* = \argmax\limits_{\pi \in \Pi} J(\pi, \mu(\nu)) \in \Pi_n \, .
    \end{align*}
    Since $\nu$ is a Nash equilibrium of the restricted game by assumption, then necessarily $J(\pi^*, \mu(\nu)) \leq J(\pi(\nu), \mu(\nu))$, and thus $\forall \pi \in \Pi, J(\pi, \mu(\nu)) \leq J(\pi(\nu), \mu(\nu))$, which concludes the proof.
\end{proof}

Using the former discussion and this property, we deduce 

\begin{theorem}[mean-field PSRO convergence to Nash equilibria]
    MF-PSRO(Nash) converges to a Nash equilibrium of the true game.
\end{theorem}

\section{Convergence to (coarse) correlated equilibria}\label{section:cce_convergence}

We now turn our attention to the versions of MF-PSRO that aim to compute mean-field correlated equilibria and mean-field coarse correlated equilibria.

\subsection{Overview}

Computing restricted MF(C)CEs is potentially more involved than computing restricted MFNE; while the optimisation problem defining restricted Nash equilibria is over the finite-dimensional space $\Delta(\Pi_n)$, the optimisation problem defining restricted MF(C)CEs is over the infinite-dimensional space $\Delta(\Delta(\Pi_n))$. One could resort to computing an approximate MFNE (a special case of both MFCE and MFCCE) using the black-box optimisation approach described in the previous section, but it is possible to exploit the structure of the mean-field game to compute approximate MF(C)CEs more efficiently. The approach we pursue is fundamentally based on no-regret learning; we also find opportunities to increase the quality of the approximate equilibrium by post-processing the output of the regret-minimisation algorithm via linear programming; see Figure~\ref{fig:mf-psro-c-ce-overview} for an overview of the techniques at play.

\subsection{Approximate (coarse) correlated equilibria via regret minimisation}\label{subsection:bandit_setting}

Our goal is to approximate an MF(C)CE for the restricted MFG based on the policy set $\Pi_n = \{ \pi_1, \ldots, \pi_n \}$, as required within the inner loop of Algorithm~\ref{alg:mf-psro-c-ce}. Recall that this amounts to solving the optimisation problem
\begin{align*}
    \displaystyle \rho_n = \argmin_{\rho \in \Delta(\Delta(\Pi_n))} \max_{i=1,\ldots,n} \mathbb{E}_{\nu \sim \rho, \pi \sim \nu}[ J(\pi_i, \mu(\nu)) - J(\pi, \mu(\nu))]
\end{align*}
in the case of coarse correlated equilibria, and 
\begin{align*}
    \displaystyle \rho_n = \argmin_{\rho \in \Delta(\Delta(\Pi_n))} \mathbb{E}_{\nu \sim \rho, \pi \sim \nu}[\max_{i=1..n} J(\pi_i, \mu(\nu)) - J(\pi, \mu(\nu))]
\end{align*}
in the case of correlated equilibria. In principle, similar black-box techniques described for approximating Nash equilibria in the previous section may be applied to solve these problems too. However, such an approach is likely to be inefficient in practice, and instead we build on regret-minimisation theory, a classical approach to computing (C)CEs in $N$ player games.

The overall approach relies on the fact that if the population distribution $\mu$ is fixed, the payoff function $\mathbb{E}_{\pi \sim \nu}[J(\pi, \mu)]$ is linear in the distribution $\nu \in \Delta(\Pi_n)$, and we are in fact considering online linear optimisation problems. Focusing first on the case of coarse correlated equilibria, we will make use of Algorithms $\mathbf{A}$ achieving $O(\sqrt{T})$ external regret in online linear optimisation, of the form described in Algorithm~\ref{alg:regret-min}.

\vspace{-0.2cm}
\begin{algorithm}%
\SetAlgoLined
\KwResult{A sequence of predictions $(\nu_t)_{t=1}^T$ such that $\max_{\nu \in \Delta(\Pi_n)} \sum_{t=1}^T R_t(\nu) - \sum_{t=1}^T R_t(\nu_t) = O(\sqrt{T})$.}
 \For{$t=1,2,\ldots,T$}{
    Algorithm makes a prediction $\nu_t \in \Delta(\Pi_n)$\;
    
    Algorithm observes a linear reward function $R_t : \Delta(\Pi_n) \rightarrow \mathbb{R}$\;
    
    Algorithm receives the reward $R_t(\nu_t)$\;
 }
 \caption{Generic form of regret-minimisation algorithm for online linear optimisation on the domain $\Delta(\Pi_n)$.}\label{alg:regret-min}
\end{algorithm}
\vspace{-0.2cm}

We may apply such an algorithm for MF(C)CE computation as shown in Algorithm~\ref{alg:regret-protocol}.

\begin{algorithm}%
\SetAlgoLined
 \For{$t=1,2,\ldots,T$}{
    Representative player selects distribution $\nu_t \in \Delta(\Pi_n)$ using a regret-minimisation algorithm $\mathbb{A}$ based on past loss function $(R_s)_{s=1}^{t-1}$\;
    
    Player observes reward function $R_t(\nu) = \mathbb{E}_{\pi \sim \nu}[J(\pi, \mu(v_t)]$\;
    
    Representative player receives reward $R_t(\nu_t) = \mathbb{E}_{\pi \sim \nu_t}[J(\pi, \mu(v_t)]$\;
 }
 Return empirical average $\rho = \frac{1}{T} \sum_{t=1}^T \delta_{\nu_t}$.
 \caption{Protocol for computing an approximate MF(C)CE via regret-minimisation}\label{alg:regret-protocol}\label{alg:meta_algo_no_regret_conversion}
\end{algorithm}

This algorithm returns the empirical average $\frac{1}{T} \sum_{t=1}^T \delta_{\nu_t}$, which is in fact an approximate MF(C)CE for the restricted game, as the following result shows.

\begin{proposition}
    The empirical average $\rho = \frac{1}{T} \sum_{t=1}^T \delta_{\nu_t}$ returned by Algorithm~\ref{alg:regret-protocol} using a regret-minimisation algorithm $\mathbb{A}$ of the form described in Algorithm~\ref{alg:regret-min}, is a $O(1/\sqrt{T})$-MF(C)CE for the restricted mean-field game.
\end{proposition}
\begin{proof}
    This is a direct computation. The benefit of the representative player deviating to $\pi_i$ under the correlation device $\rho$ is
    \begin{align*}
        & \mathbb{E}_{\nu \sim \rho}[J(\pi_i, \mu(\nu)) - \mathbb{E}_{\pi \sim \nu}[J(\pi, \mu(\nu))]] \\
        = & \frac{1}{T} \sum_{t=1}^T \left(J(\pi_i, \mu(\nu_t)) - \mathbb{E}_{\pi \sim \nu_t}[J(\pi, \mu(\nu_t))] \right) \\
        = & \frac{1}{T} O\left(\sqrt{T}\right) = O\left(1/\sqrt{T}\right) \, ,
    \end{align*}
    where the penultimate equality follows from the regret-minimising property of algorithm $\mathbb{A}$. The proof for CEs is similar.
\end{proof}

This result establishes a rigorous means of approximating an MF(C)CE in the restricted game considered within the inner loop of mean-field PSRO, and therefore provides an implementable version of mean-field PSRO. By strengthening the regret minimisation algorithm described above to minimise \emph{internal} regret, we obtain a time-average strategy that is an approximate MFCE. In both cases, we have the following correctness guarantee for MF-PSRO.

\begin{theorem}[MF-PSRO Convergence to MF(C)CEs]\label{theorem:mf_psro_convergence}
    MF-PSRO using a no-internal-regret (Respectively no-external-regret) algorithm to compute its MFCE (Respectively MFCCE) with average regret threshold $\epsilon$ and Best-Response Computation $BR_{CE}$ (Respectively $BR_{CCE}$) converges to an $\epsilon$-MFCE (Respectively an $\epsilon$-MFCCE).
\end{theorem}
\begin{proof}
    Based on previous discussions, we know that PSRO must necessarily terminate. 
    
    If PSRO terminates when using a restricted MFCCE, we must have  $$\pi^* = \argmax\limits_\pi \sum_\nu \rho(\nu) J(\pi, \mu(\nu)) \in \Pi_n\;.$$ 
    By definition of $\rho$, $\sum_\nu \rho(\nu) \Big(J(\pi^*, \mu(\nu)) - J(\pi(\nu), \mu(\nu)) \Big) \leq \epsilon$, and therefore $\forall \pi \in \Pi$, $\sum_\nu \rho(\nu) \Big(J(\pi, \mu(\nu)) - J(\pi(\nu), \mu(\nu)) \Big) \leq \epsilon$, ergo: $\rho$ is a mean-field $\epsilon$-coarse correlated equilibrium.

    The proof for mean-field correlated equilibria follows a similar line of arguments and is detailed in Appendix \ref{appendix:psro_ce_convergence_proof}.    
\end{proof}

As we will see in the next section, it is often possible to improve upon the uniform mixture of $(\nu_t)_{t=1}^T$ output by the regret-minimisation algorithm to obtain a more accurate approximation to an MF(C)CE.

\subsection{Improving the Bandit: Speed}\label{subsection:improving_bandit_speed}

\begin{figure}%
    \centering
    \includegraphics[scale=0.35]{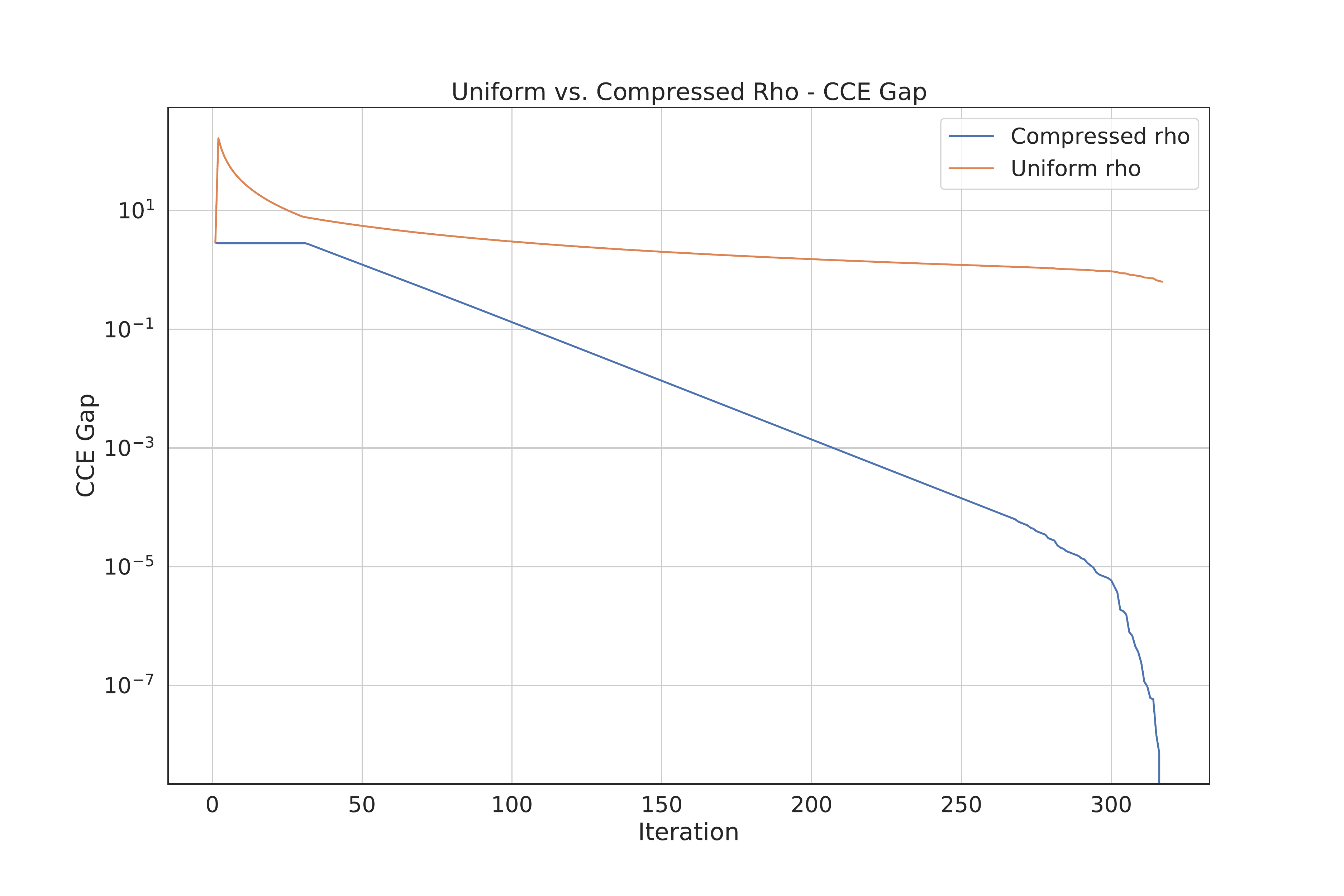}
    \caption{Uniform vs. Compressed $\rho$ - CCE Gap / Time}
    \label{fig:average_vs_compressed}
\end{figure}

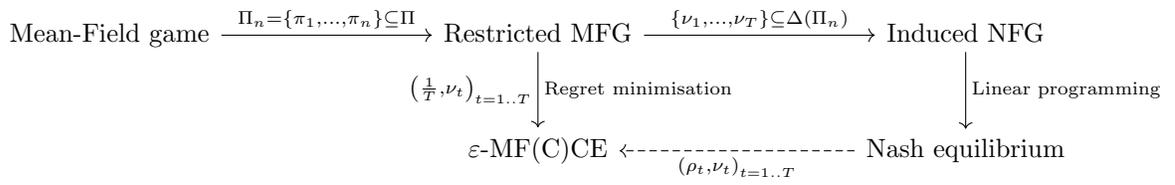
\begin{figure*}\hspace{-1.5cm}
    \begin{tikzcd}[row sep=large, column sep=8em]
        \text{Mean-Field game} \arrow{r}{\Pi_n = \{\pi_1,\ldots,\pi_n\}\subseteq \Pi} & \text{Restricted MFG} \arrow{r}{\{\nu_1,\ldots,\nu_T \} \subseteq \Delta(\Pi_n)} \arrow{d}{\text{Regret minimisation}}[swap]{\left(\frac{1}{T}, \nu_t\right)_{t=1..T}} & \text{Induced NFG} \arrow{d}{\text{Linear programming}} \\
        &  \text{$\varepsilon$-MF(C)CE}  & \arrow[dashed]{l}{\left(\rho_t, \nu_t\right)_{t=1..T}} \text{Nash equilibrium}
    \end{tikzcd}
    \caption{Reductions involved in approximation equilibrium computation in MF-PSRO.} 
    \label{fig:mf-psro-c-ce-overview}
\end{figure*}

\subsubsection{The No-Regret Speedup Algorithm: Bandit Compression}

One could use No-regret learners directly to converge towards MF(C)CE, but their equilibrium contains $T$ different distributions. This potentially means a very high amount of different $\nu_t$ recommended by our (C)CE, which can lead to learning difficulties on the part of best-responders (since every separate $\nu$ must be taken into account), implementation difficulties of equilibria in the real world, and inefficiencies: Indeed, changing per-timestep weights $\frac{1}{T}$ to potentially non-uniform $\rho_t$ can lead to converging to $\epsilon'$-MF(C)CE instead of $\epsilon$ ones, with $\epsilon' \ll \epsilon$, which is illustrated in Figure \ref{fig:average_vs_compressed}, computed at the first iteration of PSRO, in the Crowd Modelling~\cite{perrin2020fictitious} game. We define $(\rho_t)_t$ as the optimal solution of the following optimization problem:
\begin{align}
    \min\limits_{\rho} \; \max_i \; \rho^t \text{Regret}_i \label{eq:bandit_cce_optim_problem}\\
    \text{s.t.} \;\; \forall t \;\; \rho_t \geq 0, \;\; \sum_t \rho_t = 1 \nonumber
\end{align} with $\text{Regret}_i[t] := J(\pi_i, \mu(\nu_t)) - J(\pi(\nu_t), \mu(\nu_t))$.

We note that Problem (\ref{eq:bandit_cce_optim_problem}) can be interpreted as finding the row player's Nash equilibrium distribution in a zero-sum normal-form game whose payoff matrix for player 1 is Regret. We note that this objective can be expressed linearly.%

A similar problem can be solved to find better restricted mean-field correlated equilibria. First, define 
$$\text{Regret}_{i, j}(t) = \nu_t(i) \Big( J(\pi_j, \mu(\nu_t)) - J(\pi_i, \mu(\nu_t)) \Big) $$ The following problem gives optimal temporal weights $\rho$ for restricted mean-field correlated equilibria
\begin{align}
    \min\limits_{\rho} \; \max_{i, j} \; \rho^t \text{Regret}_{i, j} \label{eq:bandit_ce_optim_problem}\\
    \text{s.t.} \;\; \forall t \;\; \rho_t \geq 0, \;\; \sum_t \rho_t = 1. \nonumber
\end{align}

This problem can similarly be expressed linearly. The following theorem confirms the optimality of $\rho$, the solution of Problem (\ref{eq:bandit_cce_optim_problem}) or Problem (\ref{eq:bandit_ce_optim_problem}):

\begin{theorem}[Optimality of $\rho$]
    If $\rho = \frac{1}{T} \sum_{t=1}^T \delta_{\nu_t}$ is a restricted $\epsilon$-MFCCE (respectively $\epsilon$-MFCE), then $(\rho^*_t, \nu_t)_t$, with $\rho^*$ the optimal solution of Problem \ref{eq:bandit_cce_optim_problem} (respectively \ref{eq:bandit_ce_optim_problem}), yields a restricted $\epsilon'$-MF(C)CE of the restricted game, with $\epsilon' \leq \epsilon$; and no other $\rho$ distribution over $(\nu_t)_t$ can yield an $\epsilon''$-MF(C)CE with $\epsilon'' < \epsilon'$.
\end{theorem}
\begin{proof}
    For restricted MFCCEs, the deviation incentive against the correlation device sampling $\nu_t$ with probability $\rho_t$ in the restricted game is $$\mathbb{E}_{\nu \sim \rho, \pi \sim \nu} \left[ J(\pi', \mu(\nu)) - J(\pi, \mu(\nu)) \right] = \max\limits_i \rho^t \text{Regret}_i\,.$$
    Since the uniform distribution is a possible value for $\rho$, we necessarily have $\max\limits_i \rho^t \text{Regret}_i \leq \max\limits_i \frac{1}{T} \sum_t \text{Regret}_i[t] = \epsilon$, which concludes that part of the proof. The proof for restricted MFCEs follows the same line of arguments, and is detailed in Appendix \ref{subappendix:rho_star_optimality_ce}.
    
    Optimality of the solutions of problems (\ref{eq:bandit_cce_optim_problem}) and (\ref{eq:bandit_ce_optim_problem}) directly follows from their definitions together with the above derivations. 
\end{proof}
Given the empirical tendency of this approach to compress temporal distribution, we name it \textbf{bandit compression}. Empirically, it allows us to find much more accurate (Figure \ref{fig:average_vs_compressed}) and sparser (Appendix \ref{appendix:bandit_compression_sparsity}) distributions than uniformly averaging over $\big(\nu_t \big)_t$, and in a much lower number of steps. Yet, this algorithm is only exact in the case where the regret used by the algorithm is noiseless. The next question is therefore, how sensitive is bandit compression to noise in the regret matrix? 

\subsubsection{On the value-continuity of min-max problems}

We provide bounds on computed Average Regrets differences when $J$ is perturbed by an additive random variable $\epsilon$: $ \tilde J(\pi, \mu) =  J(\pi, \mu) + \epsilon$, giving rise to notation $\text{Regret}^\epsilon_i$, and to the identity, if we write $\tilde\epsilon_t = \epsilon_t - (\nu_t)^t \epsilon_t$, $\text{Regret}^\epsilon_i = \text{Regret}_i + \tilde\epsilon_i$.

We write \begin{align*}
    \text{Regret}_* = \min\limits_\rho \max\limits_i \rho^t \text{Regret}_i, \; \text{Regret}_*^\epsilon = \min\limits_\rho \max\limits_i \rho^t \text{Regret}^\epsilon_i
\end{align*}
We name $i_*$ and $\rho_*$ terms such that $\text{Regret}_* = (\rho_*)^t \text{Regret}_{i_*}$, and $i_*^\epsilon$ and $\rho_*^\epsilon$ the same values for $\text{Regret}_*^\epsilon$.

The quantity we wish to bound is how much additional regret we experience in expectation (ie. without noise) when using the noisy mixture weight $\rho_*^\epsilon$ instead of $\rho_*$, which we name $\Delta_O = \max\limits_i (\rho_*^\epsilon)^t \text{Regret}_{i} - (\rho_*)^t \text{Regret}_{i_*}$.

\begin{proposition}[Value-continuity of min-max] \label{proposition:value_continuity}
The optimality gap $\Delta_O$ is bounded in the following way:
$$ 0 \leq \Delta_O \leq (\rho_*)^t \tilde\epsilon_{i_*^\epsilon} - \min\limits_i(\rho_*^\epsilon)^t \tilde\epsilon_i \leq 2 ||\tilde\epsilon||_\infty \leq 4 || \epsilon ||_\infty \, .$$
\end{proposition}
\begin{proof}
    By optimality of $\rho_*$, we already have that $\Delta_O \geq 0$.
    \begin{align*}\hspace{-0.9cm}
        \Delta_O &= \max\limits_i (\rho_*^\epsilon)^t \text{Regret}_{i} - (\rho_*)^t \text{Regret}_{i_*} \\
                 &= \max\limits_i (\rho_*^\epsilon)^t (\text{Regret}_{i} + \tilde\epsilon_i) - (\rho_*^\epsilon)^t \tilde\epsilon_i - (\rho_*)^t \text{Regret}_{i_*} \\
                 &\leq (\rho_*^\epsilon)^t (\text{Regret}_{i_*^\epsilon} + \tilde\epsilon_{i_*^\epsilon}) - \min\limits_i(\rho_*^\epsilon)^t \tilde\epsilon_i - (\rho_*)^t (\text{Regret}_{i_*^\epsilon} + \tilde\epsilon_{i_*^\epsilon}) + (\rho_*)^t \tilde\epsilon_{i_*^\epsilon} \\
                 &\leq (\rho_*^\epsilon - \rho_*)^t (\text{Regret}_{i_*^\epsilon} + \tilde\epsilon_{i_*^\epsilon}) + (\rho_*)^t \tilde\epsilon_{i_*^\epsilon} - \min\limits_i(\rho_*^\epsilon)^t \tilde\epsilon_i \\
                 &\leq (\rho_*)^t \tilde\epsilon_{i_*^\epsilon} - \min\limits_i(\rho_*^\epsilon)^t \tilde\epsilon_i \leq 2 ||\tilde\epsilon||_\infty
    \end{align*}
    $\forall t, \; \tilde\epsilon_t = \epsilon_t - (\nu_t)^t \epsilon_t$, and $\epsilon_t \leq ||\epsilon||_\infty$ and $- (\nu_t)^t \epsilon_t \leq ||\epsilon||_\infty$, therefore $||\tilde\epsilon||_\infty \leq 2 ||\epsilon||_\infty$, which concludes the proof.
\end{proof}

The tightness of this bound can be verified via noting that if $\rho_* = \rho_*^\epsilon$ and the minimum of $(\rho_*)^t \epsilon_i$ is reached for $i = i_*^\epsilon$, then the optimality gap is null.

We discuss this bound in more details in Appendix \ref{appendix:continuity_minimax_examples}, where we compute its value on several examples. 

\subsubsection{The improved PSRO algorithm}

We add bandit compression onto Algorithm \ref{alg:regret-protocol}, accompanied with a few optimization criteria, yielding Algorithm \ref{alg:improved_mf_psro_ce_gen}. The improvements and their motivations are discussed in Appendix \ref{subappendix:speed_improvement_details}.

\begin{remark}[Use of the Algorithm for Nash-Convergence]\label{remark:adapted-nash-convergence}
    We note that one can also use Algorithm \ref{alg:improved_mf_psro_ce_gen} for convergence towards MFNE if one uses an iterative solver for computing the Nash equilibrium - in that case, $\mathbb{A}$ is the Nash solver, and $\text{Regret}_*$ is the exploitability. Since a Nash equilibrium only uses a single distribution, one can either bypass solving Problem \ref{eq:bandit_cce_optim_problem}, or solve it trivially with $\rho(\nu_*) = 1$.
\end{remark}

\begin{algorithm}%
\SetAlgoLined
\KwResult{Policy set  $\Pi^* = \{\pi_1, ..., \pi_n$\}, $\epsilon$-MF(C)CE $\rho^*$}
 $\Pi_0 = \emptyset$, $\Pi_1 = \{ \pi_1 \}$ with $\pi_1$ any policy, $\rho(\delta_{\pi_1}) = 1.0$, $N = 1$\;
 \While{$(\Pi_{n+1} \setminus \Pi_n) \neq \emptyset$ or $\rho_{tol} > \rho_{lim}$}{
  
  $\Pi_{n+1} = \Pi_{n} \cup \{ BR_{(C)CE}(\pi_i, \rho_T) \; | \; \pi_i, \; \rho(\pi_i) > 0 \}$ \;
  \If{$\Pi_{n+1} == \Pi_{n}$}{
    $\rho_{tol} = \frac{\rho_{tol}}{2}$
  }
  
  $n = n + 1$\;
  
  Initialize $\mathbb{A}(\Pi_{n})$\;
  
  Step Count = 0\;
  \While{$\text{Regret}_*$ > $\rho_{tol}$}{
    Step Count += 1 \;
    Do one step of $\mathbb{A}(\Pi_{n+1})$ \;
    \If{Step Count $\equiv 0 [\tau_{Compress}]$}{
        Compute $\rho_*$ optimal solution of Problem \ref{eq:bandit_cce_optim_problem} (CCE) / \ref{eq:bandit_ce_optim_problem} (CE) \;
        Compute $\rho_*$'s associated regret $\text{Regret}_*$ \;
    }
  }
  $\rho_{n+1} = \rho_*$
 }
 \caption{Sped-up mean-field PSRO((C)CE)} \label{alg:improved_mf_psro_ce_gen}
\end{algorithm}

\subsection{Complexity discussion}

The use of traditional solvers, as has been the case in PSRO so far, requires filling a payoff table. At a given iteration $n$, this means estimating $n$ match results for the newly added Best Response (The other match results being stored). 

\begin{property}[Payoff matrix estimation complexity]
    When match payoff estimation is done via sampling match outcomes, the number of matches $T$ necessary to reach within-$\epsilon$ estimation precision with probability $\alpha$ is $T = O(\frac{n}{\alpha \epsilon^2})$.
\end{property}
\begin{proof}
    If we have $T$ episodes to gather on $2n+1$ matches, the most natural (though not necessarily most efficient) way to distribute our compute budget is to give each match $\frac{T}{2n+1}$ episodes.
    
    The variance of an estimated match score $\hat J$ is therefore $\text{Var}(\hat J) = \frac{\text{Var}(J)}{\frac{T}{2n+1}} = (2n+1) \frac{\text{Var}(J)}{T}$ where $\text{Var}(J)$ is the variance of the random variable representing match outcomes for $J$.
    
    Using Chebyshev's inequality, we have $\mathbb{P}(|\hat J - J | \geq \epsilon) \leq \frac{\text{Var}(\hat J)}{\epsilon^2} = (2n+1) \frac{\text{Var}(J)}{T \epsilon^2}$. If we aim to be within $\epsilon$-precision of $J$ with probability $\alpha$, i.e. $\mathbb{P}(|\hat J - J | \geq \epsilon) = \alpha$, we need $T = O(\frac{n}{\alpha \epsilon^2})$. 
\end{proof}

We contrast this with the complexity of using no-external- and internal-regret learners, given that one chooses an efficient algorithm:

\begin{property}[Bandit $\epsilon$-Regret Complexity]
    The number of game matches $T$ necessary to reach within-$\epsilon$ average regret is $T = O\left(\frac{n^{3} \; log(n)}{\epsilon^2}\right)$ for no-internal-regret learners, and $T = O\left(\frac{n \; log(n)}{\epsilon^2}\right)$ for no-external-regret learners. In the case of additively noisy evaluation, where samples are evaluated $M$ times and averaged, these complexities become $T = O\left(\frac{M \; n^{3} \; log(n)}{\epsilon^2}\right)$ for internal-regret, and $T = O\left(\frac{n \; M \; log(n)}{\epsilon^2}\right)$ for external regret; both with probability $\delta \geq 1 - n \frac{4 \sigma^2}{T M \epsilon^2}$, where $\sigma^2$ is the noise variance.
\end{property}
\begin{proof}
    The Hedge Algorithm~\cite{blumregret} adapted for the partial-information setting~\cite{BlumInternalExternalRegret} has average regret bound $\epsilon = O\left(\sqrt{\frac{n \; log(n)}{T}}\right)$, therefore $T = O\left(\frac{n \; log(n)}{\epsilon^2}\right)$ when returns are exact.

    Optimal Swap-regret minimizers can be derived from optimal external-regret minimizers by running N instances of them in parallel, as shown in~\cite{blumregret}, therefore $\epsilon = O\left(n \sqrt{\frac{n \; log(n)}{T}}\right)$ and $T = O\left(\frac{n^{3} \; log(n)}{\epsilon^2}\right)$.
    
    The additive-payoff noise case is discussed in Appendix \ref{subappendix:proof_noisy_payoff}. The proof relies on decomposing observed regret in two terms - true-regret and noise, then applying Chebyshev's concentration inequality and bounding noise terms with $||\epsilon||_\infty$.
\end{proof}

We provide a commentary of these results in Appendix \ref{subappendix:complexity_comments}, notably comparing the traditional PSRO approach's complexity with the bandit-led approach, analysing noise sensitivity, extending them to the $N$-player case, and examining the fully-observable setting. 

\section{Experimental results} \label{section:experiments}

To demonstrate the viability of our approach, we use three different metrics presented in Section \ref{subsection:evaluation_metrics}, which we evaluate when running MF-PSRO on four different mean-field games, which are described in Section \ref{subsection:evaluation_games}. Evaluation methods are detailed in Section \ref{subsection:evaluation_methods}, and evaluation results are discussed in Section \ref{subsection:evaluation_results}.

\subsection{Evaluation metrics} \label{subsection:evaluation_metrics}

For a given correlation device $\rho$, we define $$\text{CCEGap}(\rho) := \max\limits_\pi \sum_{\nu} \rho(\nu) \big( J(\pi, \mu(\nu)) - J(\pi(\nu), \mu(\nu) \big)$$ By construction, we directly have that $\text{CCEGap}(\rho) = 0$ is equivalent to $\rho$ being an MFCCE. In the same fashion, we define $$\text{CEGap}(\rho) := \max\limits_{\pi'} \max\limits_{\pi | \rho(\pi) > 0}  \sum_{\nu} \rho(\nu | \pi) \big( J(\pi', \mu(\nu)) - J(\pi(\nu), \mu(\nu) \big)$$
for MFCE characterisation. Finally, for a given population distribution $\nu\in\Delta(\Pi)$, we introduce $$\text{Exploitability}(\nu) := \max_\pi J(\pi, \mu(\nu)) - J(\pi(\nu), \mu(\nu))$$ so that $\text{CCEGap}(\rho) = 0$, which reaches $0$ if and only if $\nu$ is an MFNE.

\subsection{Evaluation games} \label{subsection:evaluation_games}

The four games we use to evaluate convergence include atwo complex games available in OpenSpiel~\cite{lanctot2020openspiel}, Predator-Prey~\cite{perolat2021scaling} and Crowd Modeling~\cite{perrin2020fictitious}, and two new small normal-form mean-field games, \emph{Coop / Betray / Punish} and \emph{mean-field biased Rock-Paper-Scissors}, which are described in detail and motivated in Appendix \ref{subappendix:game_description_motivation}. Summarily, \emph{Coop / Betray / Punish} is a 3-action normal-form game where agents can choose to either Cooperate, and all get a good reward; betray and take advantage of others; or punish the betrayers. But punishing agents also take some reward away from cooperators (they must support the punishers). Payoffs are non-linear (quadratic) in distributions. \emph{mean-field biased Rock-Paper-Scissors} is a classic biased Rock-Paper-Scissors game, where one gets as reward for playing rock the proportion of players playing scissors minus that playing paper, all distributions multiplied by different coefficients.

\subsection{Evaluation Methods} \label{subsection:evaluation_methods}

The regret minimizer used by mean-field PSRO((C)CE) is Regret Matching~\cite{tammelin2014solving}, and the Black-Box Optimization method used by mean-field PSRO(Nash) is CMA-ES~\cite{hansen2016cma}. As per Remark \ref{remark:adapted-nash-convergence}, we use Algorithm \ref{alg:improved_mf_psro_ce_gen} for both mean-field PSRO((C)CE) and mean-field PSRO(Nash), since the Nash solver CMA-ES is iterative.

Regarding convergence to MF(C)CE, since there exists, to the best of our knowledge, no other algorithm known to converge towards these weaker equilibria  we investigate the convergence behavior of mean-field PSRO((C)CE) with additional payoff noise.

Regarding convergence towards MFNE, we compare mean-field PSRO to OMD with several different learning rates, and Fictitious Play, both algorithms available on OpenSpiel.

\subsection{Evaluation Results} \label{subsection:evaluation_results}

\begin{figure}
    \centering
    \includegraphics[scale=0.4]{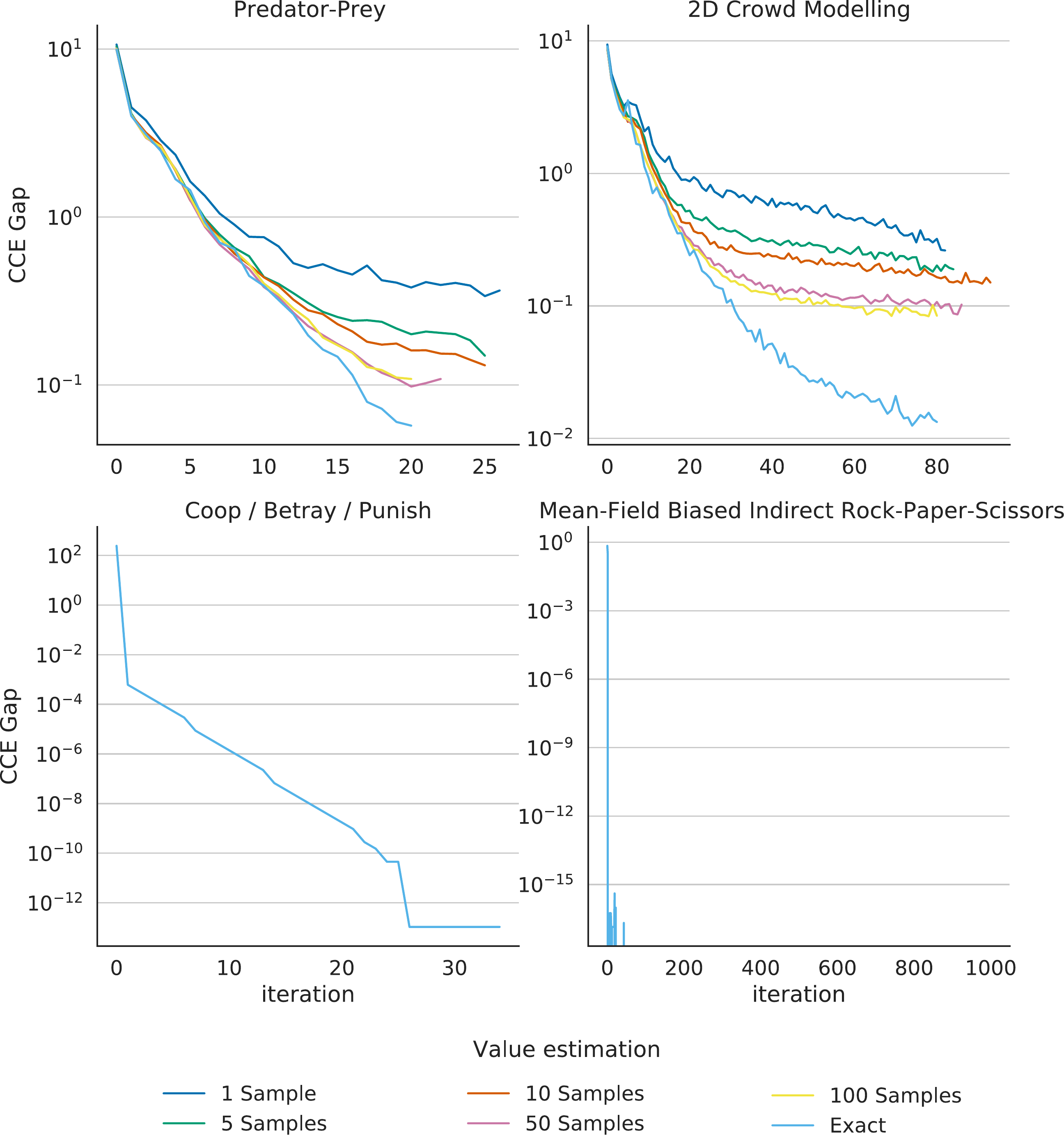}
    \caption{CCE Gap of mean-field PSRO(CCE).}
    \label{fig:cce_gap}
\end{figure}

\begin{figure}
    \centering
    \includegraphics[scale=0.4]{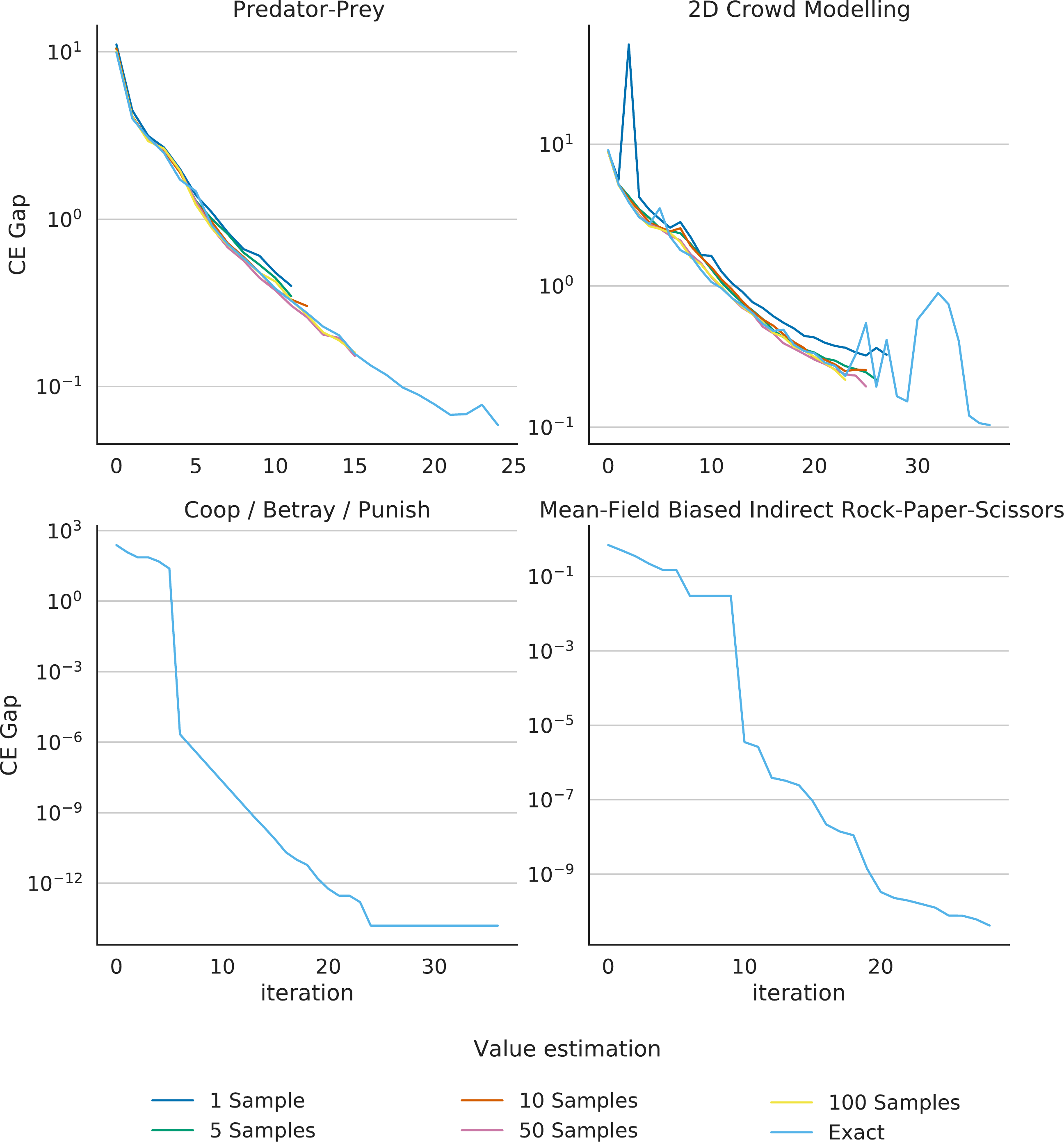}
    \caption{CE Gap of mean-field PSRO(CE).}
    \label{fig:ce_gap}
\end{figure}

\begin{figure}
    \centering
    \includegraphics[scale=0.4]{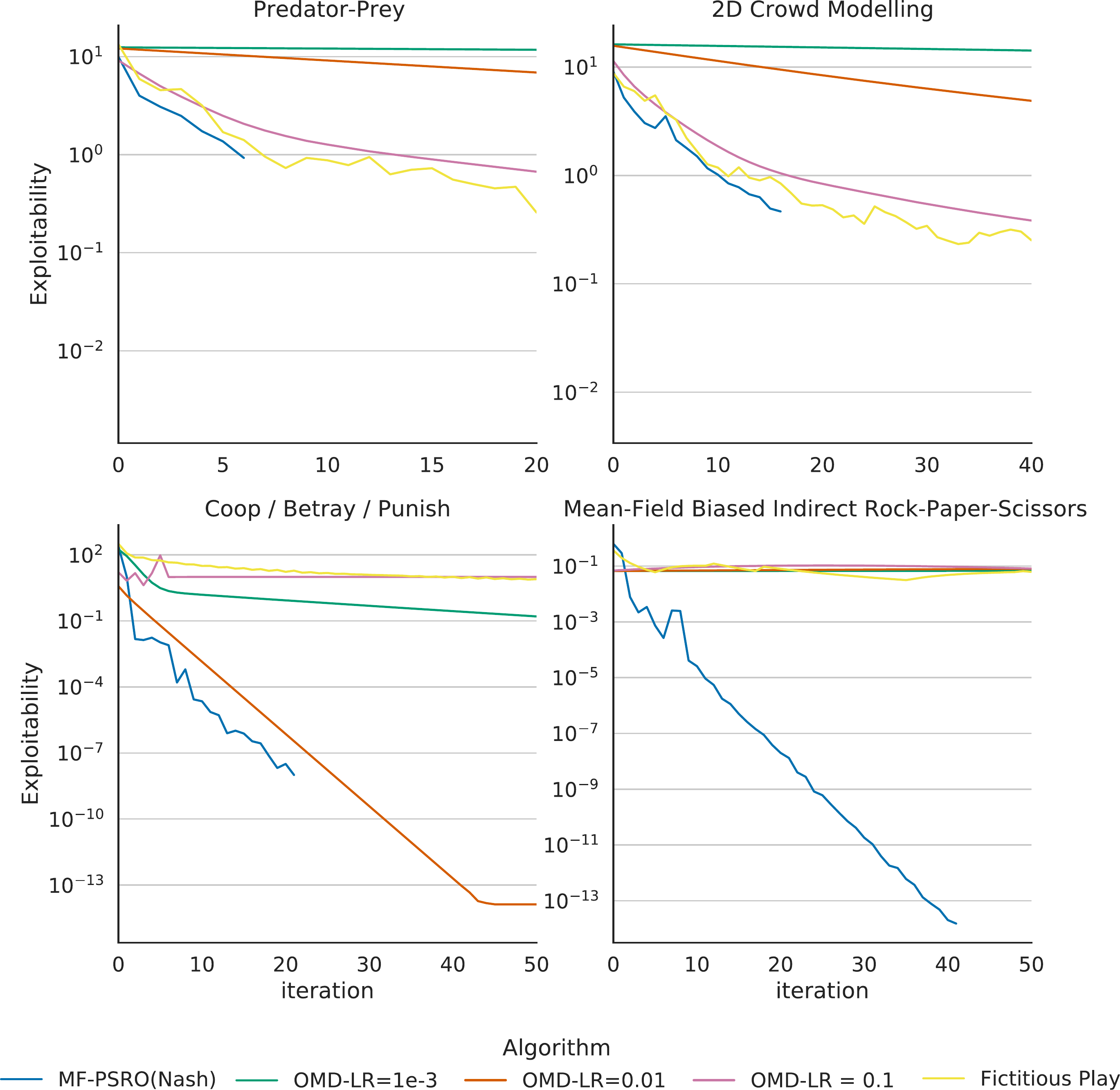}
    \caption{Exploitability of mean-field PSRO(Nash).}
    \label{fig:exploitability}
\end{figure}

Figure \ref{fig:cce_gap} presents the CCE-Gap of mean-field PSRO(CCE),  \ref{fig:ce_gap}, the CE-Gap of mean-field PSRO(CE), while Figure \ref{fig:exploitability} exposes the Exploitability of mean-field PSRO(Nash) on the four mean-field game environments described above.
We note that in both normal-form games, mean-field PSRO converges within numerical precision towards mean-field correlated, coarse correlated and Nash equilibria after only a few iterations. 

Nash-wise, OMD seems capable to follow PSRO at a similar speed on \emph{Coop / Betray / Punish}, but fails utterly to converge on \emph{mean-field biased Rock-Paper-Scissors}. We note that OMD's convergence is strongly affected by its learning rate. Fictitious play does not manage to find good equilibria in these games. 

On more complex games, mean-field PSRO quickly converges towards very good correlated ($\text{CCE Gap} \approx 10^{-1}$), coarse correlated equilibria ($\text{CE Gap} \approx 10^{-1}$), and mean-field PSRO(Nash) seems to quickly minimize exploitability - but it does much more slowly (time-wise) than both OMD and FP. This hints at a strong potential direction of improvement for mean-field PSRO. We note that in this zoomed-in plot, FP seems to outperform OMD. We provide a zoomed-out version in Appendix \ref{subappendix:zoomed_out_mf_psro_nash} where we see that OMD, with the correct learning rate, outperforms Fictitious Play as expected.

\section{Limitations}
Despite its modularity, several improvements on our approach are envisioned for further research. First, our approach cannot efficiently select higher-welfare (C)CEs over lower ones. This problem is known to be NP-Hard in general but learning approaches could hold the key to unlocking these possibilities (see the more detailed discussion in Appendix \ref{subappendix:welfare_complexity}). 
Second, mean-field PSRO(Nash) relies on a black-box algorithm, whose characteristics strongly impacts the speed and equilibrium accuracy of the algorithm. Finding a principled, general and fast Nash solver in complex restricted games, like we have for mean-field (C)CEs, could yield great improvements, both theoretically and performance-wise. 
Finally, our method is much slower than OMD or Fictitious Play on large games. This is largely due to a combination of slow payoff evaluation (be it sampled payoff or exact payoff) and relatively large amounts of steps needed to find a restricted equilibrium.

\section{Conclusion}
We have introduced a new mean-field Multi-Agent Reinforcement Learning algorithm, Mean-Field PSRO, and demonstrated its ability to converge to Nash, correlated and coarse correlated equilibria both theoretically and empirically in various benchmark games. Additionally, the approach was succesfully sped up using a new method named bandit compression, which is motivated by noise robustness and empirical speed.

The approach has only been tested so far using the computation of exact best-responses. We expect Reinforcement-Learning algorithms to work out of the box, and answering this question would unlock (C)CE convergence in very large and complex games.

\bibliographystyle{ACM-Reference-Format} 
\bibliography{references}


\begin{thebibliography}{39}


\ifx \showCODEN    \undefined \def \showCODEN     #1{\unskip}     \fi
\ifx \showDOI      \undefined \def \showDOI       #1{#1}\fi
\ifx \showISBNx    \undefined \def \showISBNx     #1{\unskip}     \fi
\ifx \showISBNxiii \undefined \def \showISBNxiii  #1{\unskip}     \fi
\ifx \showISSN     \undefined \def \showISSN      #1{\unskip}     \fi
\ifx \showLCCN     \undefined \def \showLCCN      #1{\unskip}     \fi
\ifx \shownote     \undefined \def \shownote      #1{#1}          \fi
\ifx \showarticletitle \undefined \def \showarticletitle #1{#1}   \fi
\ifx \showURL      \undefined \def \showURL       {\relax}        \fi
\providecommand\bibfield[2]{#2}
\providecommand\bibinfo[2]{#2}
\providecommand\natexlab[1]{#1}
\providecommand\showeprint[2][]{arXiv:#2}

\bibitem[\protect\citeauthoryear{Anahtarci, Kariksiz, and Saldi}{Anahtarci
  et~al\mbox{.}}{2020}]%
        {anahtarci2020q}
\bibfield{author}{\bibinfo{person}{Berkay Anahtarci}, \bibinfo{person}{Can~Deha
  Kariksiz}, {and} \bibinfo{person}{Naci Saldi}.}
  \bibinfo{year}{2020}\natexlab{}.
\newblock \showarticletitle{Q-learning in regularized mean-field games}. In
  \bibinfo{booktitle}{\emph{arXiv}}.
\newblock


\bibitem[\protect\citeauthoryear{Barman and Ligett}{Barman and Ligett}{2015}]%
        {barman2015finding}
\bibfield{author}{\bibinfo{person}{Siddharth Barman} {and}
  \bibinfo{person}{Katrina Ligett}.} \bibinfo{year}{2015}\natexlab{}.
\newblock \bibinfo{title}{Finding Any Nontrivial Coarse Correlated Equilibrium
  Is Hard}.
\newblock
\newblock
\showeprint[arxiv]{1504.06314}~[cs.GT]


\bibitem[\protect\citeauthoryear{Blum and Mansour}{Blum and Mansour}{2005}]%
        {BlumInternalExternalRegret}
\bibfield{author}{\bibinfo{person}{Avrim Blum} {and} \bibinfo{person}{Yishay
  Mansour}.} \bibinfo{year}{2005}\natexlab{}.
\newblock \bibinfo{title}{From External to Internal Regret}.
\newblock
\newblock


\bibitem[\protect\citeauthoryear{Blum and Mansour}{Blum and Mansour}{2007}]%
        {blumregret}
\bibfield{author}{\bibinfo{person}{A. Blum} {and} \bibinfo{person}{Y.
  Mansour}.} \bibinfo{year}{2007}\natexlab{}.
\newblock \showarticletitle{Learning, Regret minimization, and Equilibria}.
\newblock In \bibinfo{booktitle}{\emph{Algorithmic Game Theory}},
  \bibfield{editor}{\bibinfo{person}{Noam Nisan}, \bibinfo{person}{Tim
  Roughgarden}, \bibinfo{person}{Eva Tardos}, {and} \bibinfo{person}{Vijay~V.
  Vazirani}} (Eds.). \bibinfo{publisher}{Cambridge University Press},
  Chapter~4, \bibinfo{pages}{79--102}.
\newblock


\bibitem[\protect\citeauthoryear{Brown}{Brown}{1951}]%
        {brown1951iterative}
\bibfield{author}{\bibinfo{person}{George~W Brown}.}
  \bibinfo{year}{1951}\natexlab{}.
\newblock \showarticletitle{Iterative solution of games by fictitious play}.
\newblock \bibinfo{journal}{\emph{Activity analysis of production and
  allocation}} \bibinfo{volume}{13}, \bibinfo{number}{1}
  (\bibinfo{year}{1951}), \bibinfo{pages}{374--376}.
\newblock


\bibitem[\protect\citeauthoryear{Campi and Fischer}{Campi and Fischer}{2021}]%
        {campi2021correlated}
\bibfield{author}{\bibinfo{person}{Luciano Campi} {and} \bibinfo{person}{Markus
  Fischer}.} \bibinfo{year}{2021}\natexlab{}.
\newblock \bibinfo{title}{Correlated equilibria and mean field games: a simple
  model}.
\newblock
\newblock
\showeprint[arxiv]{2004.06185}~[math.OC]


\bibitem[\protect\citeauthoryear{Cardaliaguet and Hadikhanloo}{Cardaliaguet and
  Hadikhanloo}{2015}]%
        {cardaliaguet2015learning}
\bibfield{author}{\bibinfo{person}{Pierre Cardaliaguet} {and}
  \bibinfo{person}{Saeed Hadikhanloo}.} \bibinfo{year}{2015}\natexlab{}.
\newblock \bibinfo{title}{Learning in Mean Field Games: the Fictitious Play}.
\newblock
\newblock
\showeprint[arxiv]{1507.06280}~[math.OC]


\bibitem[\protect\citeauthoryear{Daskalakis and Papadimitriou}{Daskalakis and
  Papadimitriou}{2007}]%
        {daskalakis2007computing}
\bibfield{author}{\bibinfo{person}{Constantinos Daskalakis} {and}
  \bibinfo{person}{Christos Papadimitriou}.} \bibinfo{year}{2007}\natexlab{}.
\newblock \showarticletitle{Computing equilibria in anonymous games}. In
  \bibinfo{booktitle}{\emph{48th Annual IEEE Symposium on Foundations of
  Computer Science (FOCS'07)}}. IEEE, \bibinfo{pages}{83--93}.
\newblock


\bibitem[\protect\citeauthoryear{Daskalakis and Papadimitriou}{Daskalakis and
  Papadimitriou}{2008}]%
        {daskalakis2008discretized}
\bibfield{author}{\bibinfo{person}{Constantinos Daskalakis} {and}
  \bibinfo{person}{Christos~H Papadimitriou}.} \bibinfo{year}{2008}\natexlab{}.
\newblock \showarticletitle{Discretized multinomial distributions and Nash
  equilibria in anonymous games}. In \bibinfo{booktitle}{\emph{2008 49th Annual
  IEEE Symposium on Foundations of Computer Science}}. IEEE,
  \bibinfo{pages}{25--34}.
\newblock


\bibitem[\protect\citeauthoryear{Degl’Innocenti}{Degl’Innocenti}{2018}]%
        {deglinnocenti2018correlated}
\bibfield{author}{\bibinfo{person}{Laura Degl’Innocenti}.}
  \bibinfo{year}{2018}\natexlab{}.
\newblock \bibinfo{title}{Correlated equilibria in static mean-field games}.
\newblock
\newblock


\bibitem[\protect\citeauthoryear{Elie, Perolat, Lauri{\`e}re, Geist, and
  Pietquin}{Elie et~al\mbox{.}}{2020}]%
        {elie2020convergence}
\bibfield{author}{\bibinfo{person}{Romuald Elie}, \bibinfo{person}{Julien
  Perolat}, \bibinfo{person}{Mathieu Lauri{\`e}re}, \bibinfo{person}{Matthieu
  Geist}, {and} \bibinfo{person}{Olivier Pietquin}.}
  \bibinfo{year}{2020}\natexlab{}.
\newblock \showarticletitle{On the convergence of model free learning in mean
  field games}. In \bibinfo{booktitle}{\emph{Proceedings of the AAAI Conference
  on Artificial Intelligence}}, Vol.~\bibinfo{volume}{34}.
  \bibinfo{pages}{7143--7150}.
\newblock


\bibitem[\protect\citeauthoryear{Frazier}{Frazier}{2018}]%
        {frazier2018tutorial}
\bibfield{author}{\bibinfo{person}{Peter~I. Frazier}.}
  \bibinfo{year}{2018}\natexlab{}.
\newblock \bibinfo{title}{A Tutorial on Bayesian Optimization}.
\newblock
\newblock
\showeprint[arxiv]{1807.02811}~[stat.ML]


\bibitem[\protect\citeauthoryear{Guo, Hu, Xu, and Zhang}{Guo
  et~al\mbox{.}}{2020}]%
        {guo2020general}
\bibfield{author}{\bibinfo{person}{Xin Guo}, \bibinfo{person}{Anran Hu},
  \bibinfo{person}{Renyuan Xu}, {and} \bibinfo{person}{Junzi Zhang}.}
  \bibinfo{year}{2020}\natexlab{}.
\newblock \bibinfo{title}{A General Framework for Learning Mean-Field Games}.
\newblock
\newblock
\showeprint[arxiv]{2003.06069}~[cs.LG]


\bibitem[\protect\citeauthoryear{Hadikhanloo and Silva}{Hadikhanloo and
  Silva}{2018}]%
        {hadikhanloo2018finite}
\bibfield{author}{\bibinfo{person}{Saeed Hadikhanloo} {and}
  \bibinfo{person}{Francisco~José Silva}.} \bibinfo{year}{2018}\natexlab{}.
\newblock \bibinfo{title}{Finite mean field games: fictitious play and
  convergence to a first order continuous mean field game}.
\newblock
\newblock
\showeprint[arxiv]{1805.05940}~[math.OC]


\bibitem[\protect\citeauthoryear{Hansen}{Hansen}{2016}]%
        {hansen2016cma}
\bibfield{author}{\bibinfo{person}{Nikolaus Hansen}.}
  \bibinfo{year}{2016}\natexlab{}.
\newblock \bibinfo{title}{The CMA Evolution Strategy: A Tutorial}.
\newblock
\newblock
\showeprint[arxiv]{1604.00772}~[cs.LG]


\bibitem[\protect\citeauthoryear{Hart and Mas-Colell}{Hart and
  Mas-Colell}{2013}]%
        {hart2013simple}
\bibfield{author}{\bibinfo{person}{Sergiu Hart} {and} \bibinfo{person}{Andreu
  Mas-Colell}.} \bibinfo{year}{2013}\natexlab{}.
\newblock \bibinfo{booktitle}{\emph{Simple adaptive strategies: from
  regret-matching to uncoupled dynamics}}. Vol.~\bibinfo{volume}{4}.
\newblock \bibinfo{publisher}{World Scientific}.
\newblock


\bibitem[\protect\citeauthoryear{Huang, Malham{\'e}, Caines,
  et~al\mbox{.}}{Huang et~al\mbox{.}}{2006}]%
        {huang2006large}
\bibfield{author}{\bibinfo{person}{Minyi Huang}, \bibinfo{person}{Roland~P
  Malham{\'e}}, \bibinfo{person}{Peter~E Caines}, {et~al\mbox{.}}}
  \bibinfo{year}{2006}\natexlab{}.
\newblock \showarticletitle{Large population stochastic dynamic games:
  closed-loop McKean-Vlasov systems and the Nash certainty equivalence
  principle}.
\newblock \bibinfo{journal}{\emph{Communications in Information \& Systems}}
  \bibinfo{volume}{6}, \bibinfo{number}{3} (\bibinfo{year}{2006}),
  \bibinfo{pages}{221--252}.
\newblock


\bibitem[\protect\citeauthoryear{Kakutani}{Kakutani}{1941}]%
        {kakutani1941generalization}
\bibfield{author}{\bibinfo{person}{Shizuo Kakutani}.}
  \bibinfo{year}{1941}\natexlab{}.
\newblock \showarticletitle{A generalization of Brouwer’s fixed point
  theorem}.
\newblock \bibinfo{journal}{\emph{Duke mathematical journal}}
  \bibinfo{volume}{8}, \bibinfo{number}{3} (\bibinfo{year}{1941}),
  \bibinfo{pages}{457--459}.
\newblock


\bibitem[\protect\citeauthoryear{Khan, Zhang, Lee, Kumar, and Ribeiro}{Khan
  et~al\mbox{.}}{2018}]%
        {khan2018scalable}
\bibfield{author}{\bibinfo{person}{Arbaaz Khan}, \bibinfo{person}{Clark Zhang},
  \bibinfo{person}{Daniel~D. Lee}, \bibinfo{person}{Vijay Kumar}, {and}
  \bibinfo{person}{Alejandro Ribeiro}.} \bibinfo{year}{2018}\natexlab{}.
\newblock \bibinfo{title}{Scalable Centralized Deep Multi-Agent Reinforcement
  Learning via Policy Gradients}.
\newblock
\newblock
\showeprint[arxiv]{1805.08776}~[cs.LG]


\bibitem[\protect\citeauthoryear{Lanctot et~al\mbox{.}}{Lanctot
  et~al\mbox{.}}{2017}]%
        {lanctot2017unified}
\bibfield{author}{\bibinfo{person}{Marc Lanctot} {et~al\mbox{.}}}
  \bibinfo{year}{2017}\natexlab{}.
\newblock \bibinfo{title}{A Unified Game-Theoretic Approach to Multiagent
  Reinforcement Learning}.
\newblock
\newblock
\showeprint[arxiv]{1711.00832}~[cs.AI]


\bibitem[\protect\citeauthoryear{Lanctot et~al\mbox{.}}{Lanctot
  et~al\mbox{.}}{2020}]%
        {lanctot2020openspiel}
\bibfield{author}{\bibinfo{person}{Marc Lanctot} {et~al\mbox{.}}}
  \bibinfo{year}{2020}\natexlab{}.
\newblock \bibinfo{title}{OpenSpiel: A Framework for Reinforcement Learning in
  Games}.
\newblock
\newblock
\showeprint[arxiv]{1908.09453}~[cs.LG]


\bibitem[\protect\citeauthoryear{Lasry and Lions}{Lasry and Lions}{2007}]%
        {lasry2007mfg}
\bibfield{author}{\bibinfo{person}{Jean-Michel Lasry} {and}
  \bibinfo{person}{Pierre-Louis Lions}.} \bibinfo{year}{2007}\natexlab{}.
\newblock \showarticletitle{Mean Field Games}.
\newblock \bibinfo{journal}{\emph{Japanese Journal of Mathematics}}
  \bibinfo{volume}{2} (\bibinfo{date}{03} \bibinfo{year}{2007}),
  \bibinfo{pages}{229--260}.
\newblock
\urldef\tempurl%
\url{https://doi.org/10.1007/s11537-007-0657-8}
\showDOI{\tempurl}


\bibitem[\protect\citeauthoryear{Marris, Muller, et~al\mbox{.}}{Marris
  et~al\mbox{.}}{2021}]%
        {marris2021multiagent}
\bibfield{author}{\bibinfo{person}{Luke Marris}, \bibinfo{person}{Paul Muller},
  {et~al\mbox{.}}} \bibinfo{year}{2021}\natexlab{}.
\newblock \bibinfo{title}{Multi-Agent Training beyond Zero-Sum with Correlated
  Equilibrium Meta-Solvers}.
\newblock
\newblock
\showeprint[arxiv]{2106.09435}~[cs.MA]


\bibitem[\protect\citeauthoryear{Matignon, Laurent, and Le~Fort-Piat}{Matignon
  et~al\mbox{.}}{2012}]%
        {matignon2012independent}
\bibfield{author}{\bibinfo{person}{Laetitia Matignon},
  \bibinfo{person}{Guillaume~J Laurent}, {and} \bibinfo{person}{Nadine
  Le~Fort-Piat}.} \bibinfo{year}{2012}\natexlab{}.
\newblock \showarticletitle{Independent reinforcement learners in cooperative
  markov games: a survey regarding coordination problems}.
\newblock \bibinfo{journal}{\emph{The Knowledge Engineering Review}}
  \bibinfo{volume}{27}, \bibinfo{number}{1} (\bibinfo{year}{2012}),
  \bibinfo{pages}{1--31}.
\newblock


\bibitem[\protect\citeauthoryear{McMahan, Gordon, and Blum}{McMahan
  et~al\mbox{.}}{2003a}]%
        {mcmahan2003planning}
\bibfield{author}{\bibinfo{person}{H~Brendan McMahan},
  \bibinfo{person}{Geoffrey~J Gordon}, {and} \bibinfo{person}{Avrim Blum}.}
  \bibinfo{year}{2003}\natexlab{a}.
\newblock \showarticletitle{Planning in the presence of cost functions
  controlled by an adversary}. In \bibinfo{booktitle}{\emph{International
  Conference on Machine Learning (ICML)}}.
\newblock


\bibitem[\protect\citeauthoryear{McMahan, Gordon, and Blum}{McMahan
  et~al\mbox{.}}{2003b}]%
        {mcmahandoubleoracle}
\bibfield{author}{\bibinfo{person}{H.~Brendan McMahan},
  \bibinfo{person}{Geoffrey~J Gordon}, {and} \bibinfo{person}{Avrim Blum}.}
  \bibinfo{year}{2003}\natexlab{b}.
\newblock \bibinfo{title}{Planning in the presence of cost functions controlled
  by an adversary}.
\newblock
\newblock


\bibitem[\protect\citeauthoryear{Monnot and Piliouras}{Monnot and
  Piliouras}{2017}]%
        {monnot2017limits}
\bibfield{author}{\bibinfo{person}{Barnab{\'e} Monnot} {and}
  \bibinfo{person}{Georgios Piliouras}.} \bibinfo{year}{2017}\natexlab{}.
\newblock \showarticletitle{Limits and limitations of no-regret learning in
  games}.
\newblock \bibinfo{journal}{\emph{The Knowledge Engineering Review}}
  \bibinfo{volume}{32} (\bibinfo{year}{2017}).
\newblock


\bibitem[\protect\citeauthoryear{Muller et~al\mbox{.}}{Muller
  et~al\mbox{.}}{2020}]%
        {muller2020generalized}
\bibfield{author}{\bibinfo{person}{Paul Muller} {et~al\mbox{.}}}
  \bibinfo{year}{2020}\natexlab{}.
\newblock \bibinfo{title}{A Generalized Training Approach for Multiagent
  Learning}.
\newblock
\newblock
\showeprint[arxiv]{1909.12823}~[cs.MA]


\bibitem[\protect\citeauthoryear{Muller, Elie, Rowland, Lauriere, Perolat,
  Perrin, Geist, Piliouras, Pietquin, and Tuyls}{Muller et~al\mbox{.}}{2022}]%
        {muller2022learning}
\bibfield{author}{\bibinfo{person}{Paul Muller}, \bibinfo{person}{Romuald
  Elie}, \bibinfo{person}{Mark Rowland}, \bibinfo{person}{Mathieu Lauriere},
  \bibinfo{person}{Julien Perolat}, \bibinfo{person}{Sarah Perrin},
  \bibinfo{person}{Matthieu Geist}, \bibinfo{person}{Georgios Piliouras},
  \bibinfo{person}{Olivier Pietquin}, {and} \bibinfo{person}{Karl Tuyls}.}
  \bibinfo{year}{2022}\natexlab{}.
\newblock \bibinfo{title}{Learning Correlated Equilibria in Mean-Field Games}.
\newblock
\newblock
\urldef\tempurl%
\url{https://doi.org/10.48550/ARXIV.2208.10138}
\showDOI{\tempurl}


\bibitem[\protect\citeauthoryear{Nisan, Roughgarden, Tardos, and
  Vazirani}{Nisan et~al\mbox{.}}{2007}]%
        {nisan2007algorithmic}
\bibfield{author}{\bibinfo{person}{Noam Nisan}, \bibinfo{person}{Tim
  Roughgarden}, \bibinfo{person}{Eva Tardos}, {and} \bibinfo{person}{Vijay~V
  Vazirani}.} \bibinfo{year}{2007}\natexlab{}.
\newblock \bibinfo{booktitle}{\emph{Algorithmic game theory}}.
\newblock \bibinfo{publisher}{Cambridge University Press}.
\newblock


\bibitem[\protect\citeauthoryear{Omidshafiei, Papadimitriou, Piliouras, Tuyls,
  Rowland, Lespiau, Czarnecki, Lanctot, Perolat, and Munos}{Omidshafiei
  et~al\mbox{.}}{2019}]%
        {omidshafiei2019alpharank}
\bibfield{author}{\bibinfo{person}{Shayegan Omidshafiei},
  \bibinfo{person}{Christos Papadimitriou}, \bibinfo{person}{Georgios
  Piliouras}, \bibinfo{person}{Karl Tuyls}, \bibinfo{person}{Mark Rowland},
  \bibinfo{person}{Jean-Baptiste Lespiau}, \bibinfo{person}{Wojciech~M.
  Czarnecki}, \bibinfo{person}{Marc Lanctot}, \bibinfo{person}{Julien Perolat},
  {and} \bibinfo{person}{Remi Munos}.} \bibinfo{year}{2019}\natexlab{}.
\newblock \bibinfo{title}{$\alpha$-Rank: Multi-Agent Evaluation by Evolution}.
\newblock
\newblock
\showeprint[arxiv]{1903.01373}~[cs.MA]


\bibitem[\protect\citeauthoryear{OroojlooyJadid and Hajinezhad}{OroojlooyJadid
  and Hajinezhad}{2021}]%
        {oroojlooyjadid2021review}
\bibfield{author}{\bibinfo{person}{Afshin OroojlooyJadid} {and}
  \bibinfo{person}{Davood Hajinezhad}.} \bibinfo{year}{2021}\natexlab{}.
\newblock \bibinfo{title}{A Review of Cooperative Multi-Agent Deep
  Reinforcement Learning}.
\newblock
\newblock
\showeprint[arxiv]{1908.03963}~[cs.LG]


\bibitem[\protect\citeauthoryear{Perolat et~al\mbox{.}}{Perolat
  et~al\mbox{.}}{2021}]%
        {perolat2021scaling}
\bibfield{author}{\bibinfo{person}{Julien Perolat} {et~al\mbox{.}}}
  \bibinfo{year}{2021}\natexlab{}.
\newblock \bibinfo{title}{Scaling up Mean Field Games with Online Mirror
  Descent}.
\newblock
\newblock
\showeprint[arxiv]{2103.00623}~[cs.AI]


\bibitem[\protect\citeauthoryear{Perrin et~al\mbox{.}}{Perrin
  et~al\mbox{.}}{2020}]%
        {perrin2020fictitious}
\bibfield{author}{\bibinfo{person}{Sarah Perrin} {et~al\mbox{.}}}
  \bibinfo{year}{2020}\natexlab{}.
\newblock \bibinfo{title}{Fictitious Play for Mean Field Games: Continuous Time
  Analysis and Applications}.
\newblock
\newblock
\showeprint[arxiv]{2007.03458}~[math.OC]


\bibitem[\protect\citeauthoryear{Rowland, Omidshafiei, Tuyls, Perolat, Valko,
  Piliouras, and Munos}{Rowland et~al\mbox{.}}{2019}]%
        {NEURIPS2019_510f2318}
\bibfield{author}{\bibinfo{person}{Mark Rowland}, \bibinfo{person}{Shayegan
  Omidshafiei}, \bibinfo{person}{Karl Tuyls}, \bibinfo{person}{Julien Perolat},
  \bibinfo{person}{Michal Valko}, \bibinfo{person}{Georgios Piliouras}, {and}
  \bibinfo{person}{Remi Munos}.} \bibinfo{year}{2019}\natexlab{}.
\newblock \showarticletitle{Multiagent Evaluation under Incomplete
  Information}. In \bibinfo{booktitle}{\emph{Advances in Neural Information
  Processing Systems}}, \bibfield{editor}{\bibinfo{person}{H.~Wallach},
  \bibinfo{person}{H.~Larochelle}, \bibinfo{person}{A.~Beygelzimer},
  \bibinfo{person}{F.~d'Alch\'{e} Buc}, \bibinfo{person}{E.~Fox}, {and}
  \bibinfo{person}{R.~Garnett}} (Eds.), Vol.~\bibinfo{volume}{32}.
  \bibinfo{publisher}{Curran Associates, Inc.}
\newblock
\urldef\tempurl%
\url{https://proceedings.neurips.cc/paper/2019/file/510f2318f324cf07fce24c3a4b89c771-Paper.pdf}
\showURL{%
\tempurl}


\bibitem[\protect\citeauthoryear{Solis and Wets}{Solis and Wets}{1981}]%
        {solis1981minimization}
\bibfield{author}{\bibinfo{person}{Francisco~J Solis} {and}
  \bibinfo{person}{Roger J-B Wets}.} \bibinfo{year}{1981}\natexlab{}.
\newblock \showarticletitle{Minimization by random search techniques}.
\newblock \bibinfo{journal}{\emph{Mathematics of operations research}}
  \bibinfo{volume}{6}, \bibinfo{number}{1} (\bibinfo{year}{1981}),
  \bibinfo{pages}{19--30}.
\newblock


\bibitem[\protect\citeauthoryear{Tammelin}{Tammelin}{2014}]%
        {tammelin2014solving}
\bibfield{author}{\bibinfo{person}{Oskari Tammelin}.}
  \bibinfo{year}{2014}\natexlab{}.
\newblock \bibinfo{title}{Solving Large Imperfect Information Games Using
  CFR+}.
\newblock
\newblock
\showeprint[arxiv]{1407.5042}~[cs.GT]


\bibitem[\protect\citeauthoryear{Verma, Varakantham, and Lau}{Verma
  et~al\mbox{.}}{2020}]%
        {verma2020entropy}
\bibfield{author}{\bibinfo{person}{Tanvi Verma}, \bibinfo{person}{Pradeep
  Varakantham}, {and} \bibinfo{person}{Hoong~Chuin Lau}.}
  \bibinfo{year}{2020}\natexlab{}.
\newblock \bibinfo{title}{Entropy based Independent Learning in Anonymous
  Multi-Agent Settings}.
\newblock
\newblock
\showeprint[arxiv]{1803.09928}~[cs.LG]


\bibitem[\protect\citeauthoryear{Xie, Yang, Wang, and Minca}{Xie
  et~al\mbox{.}}{2020}]%
        {xie2020provable}
\bibfield{author}{\bibinfo{person}{Qiaomin Xie}, \bibinfo{person}{Zhuoran
  Yang}, \bibinfo{person}{Zhaoran Wang}, {and} \bibinfo{person}{Andreea
  Minca}.} \bibinfo{year}{2020}\natexlab{}.
\newblock \showarticletitle{Provable fictitious play for general mean-field
  games}.
\newblock \bibinfo{journal}{\emph{arXiv preprint arXiv:2010.04211}}
  (\bibinfo{year}{2020}).
\newblock


\end{thebibliography}

\appendix

\clearpage
\onecolumn

\section*{\centering Appendices for Learning Equilibria in Mean-Field Games: Introducting Mean-Field PSRO}

\vspace{1.5cm}

\section{Proof of graph closedness for restricted Nash existence}\label{appendix:proof_graph_closedness}

We wish to use the Kakutani Fixed Point theorem to prove the existence of restricted Nash equilibria. To do this, we have proved all the required hypothesis aside from graph-closedness of $\phi$, the best-response function. We prove here that $\text{Graph}(\phi)$ is closed.

\begin{proof}
    $\text{Graph}(\phi) = \{ (\nu, \nu') \in \Delta(\Pi_n) \times \Delta(\Pi_n) \; | \; \nu' \in \phi(\nu) \}$. Let $(\nu_k^1, \nu_k^2)_k$ be a sequence of elements of $\text{Graph}(\phi)$ which converges towards $(\nu_*^1, \nu_*^2) \in \Delta(\Pi_n) \times \Delta(\Pi_n)$.
    
    $r$ is continuous in $\mu$, therefore $J$ is also continuous in $\mu$. Since $J: (\nu_1, \nu_2) \rightarrow J(\pi(\nu_1), \mu(\nu_2))$ is linear in $\nu_1$ because $J\left(\pi(\nu_1), \mu(\nu_2)\right) = \sum_i \nu_1^i J\left(\pi_i, \mu(\nu_2)\right)$, it is also bicontinuous. 
    
    Since $J$ is bicontinuous, let $\epsilon > 0$ and $\alpha > 0$ be such that $\forall (\nu_1, \nu_2) \in \Delta(\Pi_n) \times \Delta(\Pi_n)$ such that $d\left((\nu_1, \nu_2), (\nu_*^1, \nu_*^2)\right) \leq \alpha$, 
    $$|J\left(\pi(\nu_1), \mu(\nu_2)\right) - J(\pi(\nu_*^1), \mu(\nu_*^2))| \leq \epsilon$$
    with $d$ a metric over $\Delta(\Pi_n) \times \Delta(\Pi_n)$ under which $J$ is continuous. Let $N_0 > 0$ be such that $\forall n \geq N_0, \; d\left((\nu_k^1, \nu_k^2), (\nu_*^1, \nu_*^2)\right) \leq \alpha$, and let $n \geq N_0$.
    
    By bicontinuity and triangle inequality, $$J\big(\pi(\nu), \mu(\nu_*^2)\big) \leq \epsilon + J\big(\pi(\nu), \mu(\nu_n^2)\big)$$ 
    $$- J\big(\pi(\nu_*^1), \mu(\nu_*^2)\big) \leq \epsilon - J\big(\pi(\nu_n^1), \mu(\nu_n^2)\big)$$ 
    And by optimality of $\nu_n^1$ against $\mu(\nu_n^2)$, $\forall \nu \in \Delta(\Pi_n)$, $$J\big(\pi(\nu), \mu(\nu_n^2)\big) - J\big(\pi(\nu_n^1), \mu(\nu_n^2)\big) \leq 0$$
    We then have, $\forall \nu \in \Delta(\Pi_n)$, 
    \begin{align*}\hspace{-0.2cm}
        J(\pi(\nu), \mu(\nu_*^2)) - J(\pi(\nu_*^1), \mu(\nu_*^2)) &\leq 2\epsilon + J(\pi(\nu), \mu(\nu_n^2)) - J(\pi(\nu_n^1), \mu(\nu_n^2)) \\
        &\leq 2 \epsilon        
    \end{align*}
    This is true for all $\nu$, so also for their sup:
    $$\sup_\nu J(\pi(\nu), \mu(\nu_*^2)) - J(\pi(\nu_*^1), \mu(\nu_*^2)) \leq 2\epsilon$$
    Finally, this is true for all $\epsilon > 0$. Taking $\epsilon$ to 0, we have that $J(\pi(\nu_*^1), \mu(\nu_*^2)) = \sup_\nu J(\pi(\nu), \mu(\nu_*^2))$, and thus $(\nu_*^1, \nu_*^2) \in \text{Graph}(\phi)$. Therefore $\text{Graph}(\phi)$ is closed.
\end{proof}

\section{Proof of convergence of mean-field PSRO towards mean-field correlated equilibria} \label{appendix:psro_ce_convergence_proof}

\begin{proof}
    If PSRO terminates when using a restricted mean-field correlated equilibrium, then it means that $\forall \pi_k, \; \rho(\pi_k) > 0, \; \pi^*(\pi_k) = \argmax\limits_\pi \sum_\nu \rho(\nu | \pi_k) J(\pi, \mu(\nu)) \in \Pi_n$. By definition of $\rho$, $\forall \pi \in \Pi_n, \; \rho(\pi) \sum_\nu \rho(\nu|\pi) \Big(J(\pi^*(\pi), \mu(\nu)) - J(\pi, \mu(\nu)) \Big) \leq \epsilon$, and therefore $\forall \pi' \in \Pi$, $\sum_\nu \rho(\nu) \Big(J(\pi', \mu(\nu)) - J(\pi, \mu(\nu)) \Big) \leq \epsilon$, ergo: $\rho$ is a mean-field $\epsilon$-correlated equilibrium.
\end{proof}

\section{The Linear special case} \label{appendix:linear_special_case}

\subsection{Definitions}

\begin{definition}[Diff-Affinity]
    We say that a function $f: x, z \rightarrow f(x, z)$ is diff-affine in z, or z-diff-affine, if $\forall x, y, \; \Delta_{x, y}(f): z \rightarrow f(x, z) - f(y, z)$ is affine.
\end{definition}

\begin{property} \label{property:diff_convexity_reward_condition}
    If $r$ is of the form $r(x, a, \mu) = C(\mu) + r_1(x, a)^t \mu + r_2(x, a)$, with $C$ any function of $\mu$, then J is diff-affine in $\mu$. Provided $r$ is $\mathcal{C}^2$, this property is also necessary.
\end{property}

We note that our following proofs' logic can also be applied with an approximate version of diff-affinity, where $$J(\pi', \mu(\nu)) - J(\pi, \mu(\nu)) \leq \sum_\pi \nu(\pi) \left( J(\pi', \mu^\pi) - J(\pi, \mu^\pi) \right) + \epsilon$$ this is for example the case when $r = f + g$, with $f$ a diff-affine function in $\mu$, and $\forall (x, a, \mu, \mu'), \; |g(x, a, \mu') - g(x, a, \mu)| \leq \epsilon$. In this case, mean-field PSRO converges to $\epsilon$ variants of our equilibria. 

\begin{remark}
    Requiring that $f: x, z \rightarrow f(x, z)$ to be such that $\forall x, y, \; \Delta_{x, y}(f): z \rightarrow f(x, z) - f(y, z)$ is convex (ie. is diff-convex) is equivalent to requiring that $f$ be diff-affine.
\end{remark}

\begin{proof}
    Let $f$ be diff-convex. Then we know that, since $f$ is scalar, $\Delta_{x, y}(f)$ is as well for all $x, y$. If $f$ is twice-differentiable in $z$, so is $\Delta_{x, y}(f)$ for all values of $x, y$. 
    The convexity condition on $\Delta_{x, y}(f)$ can be rewritten, if $f$ is twice-differentiable, as 
    
    \begin{align*}
        \forall x, y, &\frac{d^2 \Delta_{x, y}(f)}{dz^2} \geq 0 \\
        &\frac{d^2 f(x, z)}{dz^2} \geq \frac{d^2 f(y, z)}{dz^2} \\
    \end{align*}
    Inverting $x$ and $y$, we find that we have necessarily, $\forall x, y, \frac{d^2 f(x, z)}{dz^2} = \frac{d^2 f(y, z)}{dz^2} = c(z)$, therefore we know that $\forall x, z, f(x, z) = C(z) + a(x) z + b(x)$
\end{proof}

\subsection{Normal-form games equilibria and links to mean-field restricted games}

This section presents results linking restricted games with normal-form representations under the $\mu$-Diff-Affinity condition. We name $\Pi_n = \{ \pi_1, ..., \pi_n \}$ the set of policies used by the restricted game in the following.

\subsubsection{Nash equilibrium}

We wish to compute the restricted mean-field Nash equilibrium of given policies $\pi_1, ..., \pi_n$. To do this, we store values $(J(\pi_i, \mu^{\pi_j}))_{i, j}$ in a payoff matrix and compute the Nash equilibrium of the two-player game defined as follows : Player 1 receives the payoff received when player 1 chooses a deviating policy $i$ and player 2 chooses the population-generating policy $j$; Player 2 receives the transposed payoff (ie. Player 1 picks the population-generating policy $i$ and Player 2 picks the deviating policy $j$).

\[
\hspace{-0.4cm}
\begin{pmatrix}
J(\pi_1, \mu^{\pi_1}) & ... & J(\pi_1, \mu^{\pi_n})\\
J(\pi_2, \mu^{\pi_1}) & ... & J(\pi_2, \mu^{\pi_n})\\
... & ... & ...\\
J(\pi_{n-1}, \mu^{\pi_1}) & ... & J(\pi_{n-1}, \mu^{\pi_{n}})\\
J(\pi_n, \mu^{\pi_1}) & ... & J(\pi_n, \mu^{\pi_n})
\end{pmatrix}
,
\begin{pmatrix}
J(\pi_1, \mu^{\pi_1}) & ... & J(\pi_n, \mu^{\pi_1})\\
J(\pi_1, \mu^{\pi_2}) & ... & J(\pi_n, \mu^{\pi_2})\\
... & ... & ...\\
J(\pi_1, \mu^{\pi_{n-1}}) & ... & J(\pi_{n}, \mu^{\pi_{n-1}})\\
J(\pi_1, \mu^{\pi_n}) & ... & J(\pi_n, \mu^{\pi_n})
\end{pmatrix}
\]

\begin{theorem}[Normal-form and restricted game equivalence]\label{theorem:nf_rg_nash_eq}
    If $J$ is $\mu$-diff-affine, then any symmetric Nash equilibrium of the symmetric two-player game defined above is also a Nash-equilibrium of the restricted mean-field game defined by $\pi_1, ..., \pi_n$. 
\end{theorem}
\begin{proof}
    Let $\nu$ be a symmetric Nash equilibrium of the normal-form game. Then we have that
    \begin{align*}
        \forall \pi' \in \{ \pi_1, ..., \pi_n \}, \sum_i \sum_j \nu_i \nu_j \left( J(\pi', \mu^{\pi_j}) - J(\pi_i, \mu^{\pi_j}) \right) &\leq 0 \\
        \sum_i \nu_i \Delta_{\pi', \pi_i}(J) \underbrace{\left(\sum_j \nu_j \mu^{\pi_j}\right)}_{= \mu(\nu)}  \leq \sum_i \sum_j \nu_i \nu_j \Delta_{\pi, \pi_i}(J)(\mu^{\pi}_j) &\leq 0 \\ 
        J(\pi, \mu(\nu)) - \sum_i \nu_i J(\pi_i, \mu(\nu)) &\leq 0 \\
        J(\pi, \mu(\nu)) - J(\pi(\nu), \mu(\nu)) &\leq 0 \\
    \end{align*}
    where the last line comes from the fact that $\pi(\nu)$ is exactly the policy resulting from sampling $\pi$ from $\nu$ at the start of every episode.
    Therefore $\pi(\nu)$ is a Nash equilibrium of the game if we restrict deviations to be within the set of the $(\pi_i)_i$.
\end{proof}

We must note one important corollary: since the Nash equilibrium of a $\mu$-diff-affine restricted game can be expressed as the symmetric Nash equilibrium of a 2-player symmetric normal-form game, then, according to the Nash Theorem, this Nash equilibrium always exists. This, in turn, guarantees the existence of correlated and coarse-correlated equilibria in $\mu$-diff-affine games. 

\begin{corollary}[Restricted game equilibrium existence]
    In a $\mu$-diff-affine restricted game, Nash, correlated and coarse-correlated equilibria always exist.
\end{corollary}

\subsubsection{Restricted-game coarse correlated equilibrium}

We define a restricted game coarse correlated equilibrium as a recommendation device $\rho$ which recommends population distributions $\nu \in \Delta(\Pi_n)$ such that $$ \max_{\pi_k} \sum_{\nu} \rho(\nu) \sum_i \sum_j \nu_i \nu_j \left(J(\pi_k, \mu_j) - J(\pi_i, \mu_j) \right) \leq 0$$ 

$$ \text{i.e. } \; \max_{e_k} \sum_{\nu} \rho(\nu) \left( e_k - \nu \right)^t J \; \nu \leq 0$$ 

We note that although the set $\Delta(\Pi)$ is not discrete in general, the above equation is written using a sum (Though we note it could also be written using an integral), the reason being algorithmic. Indeed, given $K$ the number of searched different $\nu$ with non-zero support in $\rho$, our optimization process searches for $K$ different $\nu_k \in \mathbb{R}^N$, and their distribution $\rho \in \mathbb{R}^K$, instead of searching over the infinite-dimensional space $\mathcal{P}(\Delta(\Pi))$. 

We propose below a maximum-margin solution with parameter $K$, which one can solve using Quadratic Programming by introducing intermediary variables
\begin{align*}
    \text{Objective:}& \;\;\;\; \min_{\rho, \nu_1, ..., \nu_K} \max_{e_l} \sum_{k} \rho_k \left( e_l - \nu_k \right)^t J \; \nu_k \\
    \text{Probability constraint:}& \;\;\;\;   \;\;\;\;\;\;\;\;\; \rho \geq 0, \;\; \underline{1}^t \rho = 1 \\
                                  & \;\;\;\; \forall k, \;\; \nu_k \geq 0, \;\;\underline{1}^t \nu_k = 1 \nonumber
\end{align*}

Though note that other objectives and formulations are possible, for example a maximum-entropy one. We note that the CCE constraint is quadratic, therefore QCP-capable solvers are required, though expect usual simplifications to hold.

\begin{align}
    \text{Objective:}& \;\;\;\; \min_{\rho, \nu_1, ..., \nu_K} \sum_{k=1}^K \rho_k \; \nu_k^t \log(\nu_k) \\
    \text{CCE-Constraint:}& \;\;\;\;  \forall l, \;\; \sum_{k} \rho_k \left( e_l - \nu_k \right)^t J \; \nu_k \leq 0 \\
    \text{Probability constraint:}& \;\;\;\;   \;\;\;\;\;\;\;\;\; \rho \geq 0, \;\; \underline{1}^t \rho = 1 \\
                                  & \;\;\;\; \forall k, \;\; \nu_k \geq 0, \;\;\underline{1}^t \nu_k = 1 \nonumber
\end{align}

Or maximum-entropy with KL-regularization imposing differences between population recommendations, for $0 \leq \lambda \leq 1$,

\begin{align}
    \text{Objective:}& \;\;\;\; \min_{\rho, \nu_1, ..., \nu_K} \sum_{k=1}^K \rho_k \left( \nu_k^t \log(\nu_k) - \lambda \sum_{k'=1}^K \nu_k^t \log(\frac{\nu_k}{\nu_{k'}}) \right) \\
    \text{CCE-Constraint:}& \;\;\;\;  \forall l, \;\; \sum_{k} \rho_k \left( e_l - \nu_k \right)^t J \; \nu_k \leq 0 \\
    \text{Probability constraint:}& \;\;\;\;   \;\;\;\;\;\;\;\;\; \rho \geq 0, \;\; \underline{1}^t \rho = 1 \\
                                  & \;\;\;\; \forall k, \;\; \nu_k \geq 0, \;\;\underline{1}^t \nu_k = 1 \nonumber
\end{align}

\subsubsection{Restricted-game correlated equilibrium}

We define a restricted game correlated equilibrium as a recommendation device $\rho$ which recommends population distributions $\nu \in \Delta(\Pi_n)$ such that 

\begin{align*}
\sum_i \max_{\pi_k} \sum_{\nu} \rho(\nu)  \sum_j \nu_i \nu_j \left(J(\pi_k, \mu_j) - J(\pi_i, \mu_j) \right) \leq 0 \\
\sum_i \max_{k} \sum_{\nu} \rho(\nu)  \nu_i \left(e_k - e_i \right)^t J \; \nu \leq 0 \\
\text{And thus, we have} \\
\forall i, k, \;\; \sum_{\nu} \rho(\nu)  \nu_i \left(e_k - e_i \right)^t J \; \nu \leq 0 \\
\end{align*}

As before, given a fixed number of different population distributions $K$, we suggest three different optimization objectives : A maximum-margin, quadratic optimization one

\begin{align*}
    \text{Objective:}& \;\;\; \min_{\rho, \nu_1, ..., \nu_K} \sum_i \max_{k} \sum_{\nu} \rho(\nu)  \nu_i \left(e_k - e_i \right)^t J \; \nu \\
    \text{Prob. constraint:}& \;\;\;   \;\;\;\;\;\;\;\;\; \rho \geq 0, \;\; \underline{1}^t \rho = 1 \\
                            & \;\;\; \forall k, \;\; \nu_k \geq 0, \;\;\underline{1}^t \nu_k = 1 \nonumber
\end{align*}

A maximum-entropy one

\begin{align}
    \text{Objective:}& \;\;\;\; \max_{\rho, \nu_1, ..., \nu_K} \sum_{k=1}^K \rho_k \; \nu_k^t \log(\nu_k) \\
    \text{CE-Constraint:}& \;\;\;\; \forall i, k, \;\; \sum_{\nu} \rho(\nu)  \nu_i \left(e_k - e_i \right)^t J \; \nu \leq 0 \\
    \text{Probability constraint:}& \;\;\;\;   \;\;\;\;\;\;\;\;\; \rho \geq 0, \;\; \underline{1}^t \rho = 1 \\
                                  & \;\;\;\; \forall k, \;\; \nu_k \geq 0, \;\;\underline{1}^t \nu_k = 1 \nonumber
\end{align}

Or a maximum-entropy with KL-regularization imposing differences between population recommendations, when $0 \leq \lambda \leq 1$

\begin{align}
    \text{Objective:}& \;\;\;\; \max_{\rho, \nu_1, ..., \nu_K} \sum_{k=1}^K \rho_k \left( \nu_k^t \log(\nu_k) - \lambda \sum_{k'=1}^K \nu_k^t \log(\frac{\nu_k}{\nu_{k'}}) \right) \\
    \text{CE-Constraint:}& \;\;\;\;  \forall i, k, \;\; \sum_{\nu} \rho(\nu)  \nu_i \left(e_k - e_i \right)^t J \; \nu \leq 0 \\
    \text{Probability constraint:}& \;\;\;\;   \;\;\;\;\;\;\;\;\; \rho \geq 0, \;\; \underline{1}^t \rho = 1 \\
                                  & \;\;\;\; \forall k, \;\; \nu_k \geq 0, \;\;\underline{1}^t \nu_k = 1 \nonumber
\end{align}

\subsection{Mean-Field PSRO: Convergence to Nash equilibria in diff-affine games}

MF-PSRO is defined in a very similar way to standard PSRO in diff-affine games: start with a restricted policy set $\Pi_0$, and, at each step $N$, compute the $\Pi_n$ restricted Nash equilibrium $\nu_n$, and compute a best response $\pi_{n+1}$ to this $\Pi_n$ mixed according to $\nu_n$. If $\pi_{n+1} \in \Pi_{n}$, then $\Pi_n$ mixed according to $\nu_n$ is a Nash equilibrium of the true game, otherwise the algorithm continues.

\begin{algorithm}%
\SetAlgoLined
\KwResult{Policy set  $\Pi^* = \{\pi_1, ..., \pi_n$\}, Policy Distribution $\nu^*$ yielding game Nash $\pi(\nu^*)$}
 $\Pi_1 = \{ \pi_1 \}$ with $\pi_1$ any policy, $\nu(\pi_1) = 1.0$, $N = 1$\;
 \While{$V^{BR(\mu^{\pi(\nu)}), \mu^{\pi(\nu)}} > V^{\pi(\nu), \mu^{\pi(\nu)}}$}{
  $\Pi_{n+1} = \Pi_{n} \cup \{ BR(\mu^{\pi(\nu)}) \}$ \;
  
  $n = n + 1$\;
  
  $\forall i, j \leq n, M_{i, j} = \mathbb{E}[J(\pi_i, \mu^{\pi_j})]$ \;
  
  $\nu = \text{Matrix Nash Solver}([M, M^t])$
 }
 \caption{MF-PSRO(Nash) (Diff-Convex case)}
\end{algorithm}

\subsection{Mean-Field PSRO: Convergence to (coarse) correlated equilibria in diff-affine games}

When the game is $\mu$-diff-affine, we have the following property

\begin{property} \label{}
    In a $\mu$-diff-affine game, any (coarse) correlated equilibrium of the restricted game is a (coarse) correlated equilibrium of the True game when deviations are restricted to the set of known policies $(\pi_n)_n$
\end{property}

\begin{proof}
    Since the game is $\mu$-diff-affine, for all $\nu \in \Delta(\Pi_n)$, $\pi_k \in \Pi_n$ we have $$ J(\pi_k, \mu(\nu)) - J(\pi_i, \mu(\nu)) =  \sum_j \nu_j \left( J(\pi_k, \mu_j) - J(\pi_i, \mu_j) \right) $$
    
    Let $\rho$ be the correlation device of an MFCE of the restricted game. Then 
    
    \begin{align*}
        \sum_{i} \max_{\pi_k \in \Pi_n} \sum_{\nu} \rho(\nu) \nu_i \underbrace{\sum_j \nu_j \left( J(\pi_k, \mu_j) - J(\pi_i, \mu_j) \right)}_{\geq J(\pi_k, \mu(\nu)) - J(\pi_i, \mu(\nu))} &\leq 0 \\
        \sum_{i} \max_{\pi_k \in \Pi_n} \sum_{\nu} \rho(\nu) \nu(\pi_i) \left( J(\pi_k, \mu(\nu)) - J(\pi_i, \mu(\nu)) \right) &\leq 0
    \end{align*}
    therefore $\rho$ is a mean-field correlated equilibrium.
    
    Let $\rho$ be the correlation device of an MFCCE of the restricted game. Then 
    \begin{align*}
        \max_{\pi_k \in \Pi_n} \sum_{i} \sum_{\nu} \rho(\nu) \nu_i \underbrace{\sum_j \nu_j \left( J(\pi_k, \mu_j) - J(\pi_i, \mu_j) \right)}_{\geq J(\pi_k, \mu(\nu)) - J(\pi_i, \mu(\nu))} &\leq 0 \\
        \max_{\pi_k \in \Pi_n} \sum_{\nu} \rho(\nu) \left( J(\pi_k, \mu(\nu)) - J(\pi(\nu), \mu(\nu)) \right) &\leq 0
    \end{align*}
\end{proof}

Our algorithm is thus 

\begin{algorithm}%
\SetAlgoLined
\KwResult{Policy set  $\Pi^* = \{\pi_1, ..., \pi_n$\}, Policy Distribution $\nu^*$ yielding game Nash $\pi(\nu^*)$}
 $\Pi_1 = \{ \pi_1 \}$ with $\pi_1$ any policy, $\nu(\pi_1) = 1.0$, $N = 1$\;
 \While{$V^{BR_{(C)CE}(\mu^{\pi(\nu)}), \mu^{\pi(\nu)}} > V^{\pi(\nu), \mu^{\pi(\nu)}}$}{
  $\Pi_{n+1} = \Pi_{n} \cup \{ BR_{(C)CE}(\mu^{\pi(\nu)}) \}$ \;
  
  $N = N + 1$\;
  
  $\forall i, j \leq N, M_{i, j} = \mathbb{E}[J(\pi_i, \mu^{\pi_j})]$ \;
  
  $\nu = \text{Restricted-game mean-field (C)CE Solver}(M)$
 }
 \caption{MF-PSRO((C)CE) (Diff-Convex case)}\label{alg:mf_psro_cce_diff_conv}
\end{algorithm}

\section{Bandit Compression Sparsity}\label{appendix:bandit_compression_sparsity}

Figure \ref{fig:compression_sparsity} illustrates the sparsity of the bandit compression distribution, computed on the same example as figure \ref{fig:average_vs_compressed}.
\begin{figure}
    \centering
    \includegraphics[scale=0.4]{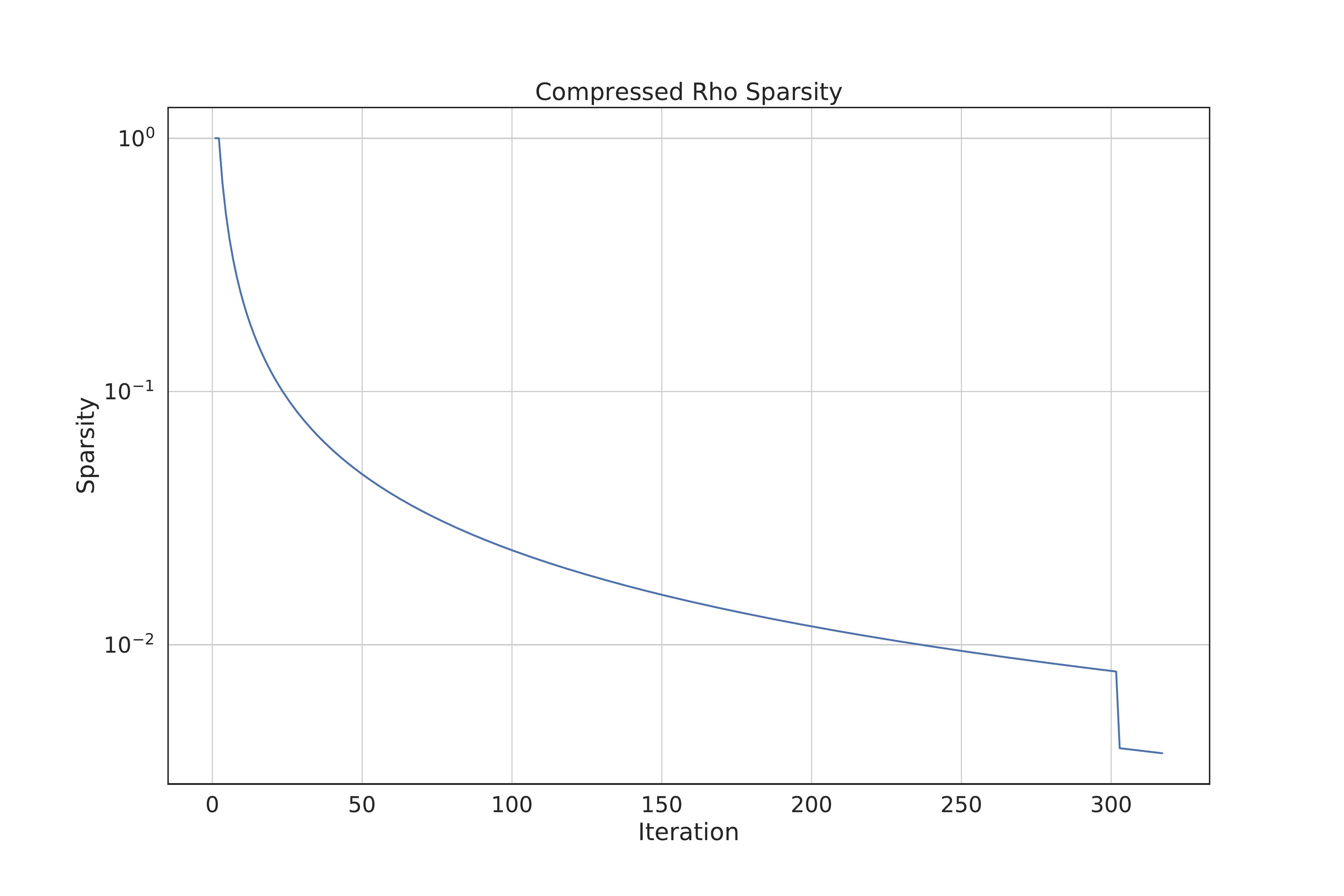}
    \caption{Bandit compression: sparsity / time.}
    \label{fig:compression_sparsity}
\end{figure}

\section{Value-Continuity of Minimax : Examples}\label{appendix:continuity_minimax_examples}

Assuming each $\epsilon_i$ is a random gaussian variable with variance $\sigma > 0$, the term $\max_i \rho^t_\Delta \epsilon_i$ is such that $\mathbb{P}(\max_i \rho^t_\Delta \epsilon_i \leq y) = \Phi^{n} \left( \frac{y}{\sigma \sqrt{\rho_\Delta^t \rho_\Delta}} \right)$ where $\Phi$ is a standard gaussian CDF, and the term on the right $\rho^t_\epsilon \epsilon_{i_\Delta} \sim \mathcal{n}(0, \frac{\sigma^2}{\rho^t_\epsilon \rho_\epsilon})$.

To get an estimation of the magnitude of this gap's distribution, we assume $N$ to be high enough that we can ignore the term $\rho^t_\epsilon \epsilon_{i_\Delta}$ for our numerical application.
Using the $5 \sigma$ rule, we find that $\mathbb{P} \left(\Delta_O \leq 5\sigma\sqrt{\rho_\Delta^t \rho_\Delta} \right ) \geq 0.9999994^N $. Assuming $\sigma = 0.1$, $N = 50$ and $\rho_\Delta^t \rho_\Delta = \frac{1}{50}$ (Fully uniform distribution), $\mathbb{P} \left(\Delta_O \leq 0.07 \right ) \geq 0.99997$. When $\rho_\Delta$ is fully focused on one point, $\rho_\Delta^t \rho_\Delta = 1$, and the former equation becomes $\mathbb{P} \left(\Delta_O \leq 0.5 \right ) \geq 0.99997$

We note that this is a pessimistic estimate for several reasons \begin{itemize}
    \item $\rho$ should presumably not be focused on a single point, and therefore the term $5 \sigma \sqrt{\rho^t \rho}$ will be quite lower
    \item This does not take into account the complex relationship and dependence between the two terms $\max_i \rho_\Delta^t \epsilon_i$ and $\rho^t_\epsilon \epsilon_{i_\Delta}$
\end{itemize}

\section{Experiments: Details and Additional Results}\label{appendix:experiment_details}

\subsection{Game description and motivation}\label{subappendix:game_description_motivation}

Both games have 3 actions, $A$, $B$ and $C$, whose rewards depend on the action distribution of the population. 
For \emph{mean-field Biased Indirect Rock Paper Scissors}, we have
$$r(A, \mu) = 0.5 * \mu(B) - 0.3 * \mu(C)$$
$$r(B, \mu) = 0.3 * \mu(C) - 0.7 * \mu(A)$$
$$r(C, \mu) = 0.7 * \mu(A) - 0.5 * \mu(B)$$

For \emph{Coop / Betray / Punish}, we have 
$$r(A, \mu) = \mu(A) - \frac{20}{9} (\mu(A) - \mu(C)) * \mu(C) - 2 \mu(B)$$
$$r(B, \mu) = 2 (\mu(A) - \mu(B)) - 238 \mu(C)$$
$$r(C, \mu) = \frac{200}{9} (\mu(A) - \mu(C)) * \mu(C)$$

\subsection{Sped-up mean-field PSRO parameter effects}

\begin{center}
    \begin{tabular}{ c | c }
         PSRO Parameter & Effect when Increased \\
         \hline $T$ & Improved asymptotic convergence, lowers speed \\  
         $M$ & Lower noise, improves convergence at fixed $T$, lower speed \\
         $\rho_{\text{Tol}}$ & Lower precision, higher speed \\
         $\tau_{Compress}$ & Higher speed if costly compression, otherwise lower \\
    \end{tabular}
\end{center}

\subsection{MF-PSRO(Nash)} \label{subappendix:zoomed_out_mf_psro_nash}

Figure \ref{fig:exploitabilities_zoomed_out} shows exploitability of MF-PSRO(Nash), OMD for several learning rates, and Fictitious Play.

\begin{centering}
\begin{figure}\hspace{-1.5cm}
    \includegraphics[scale=0.5]{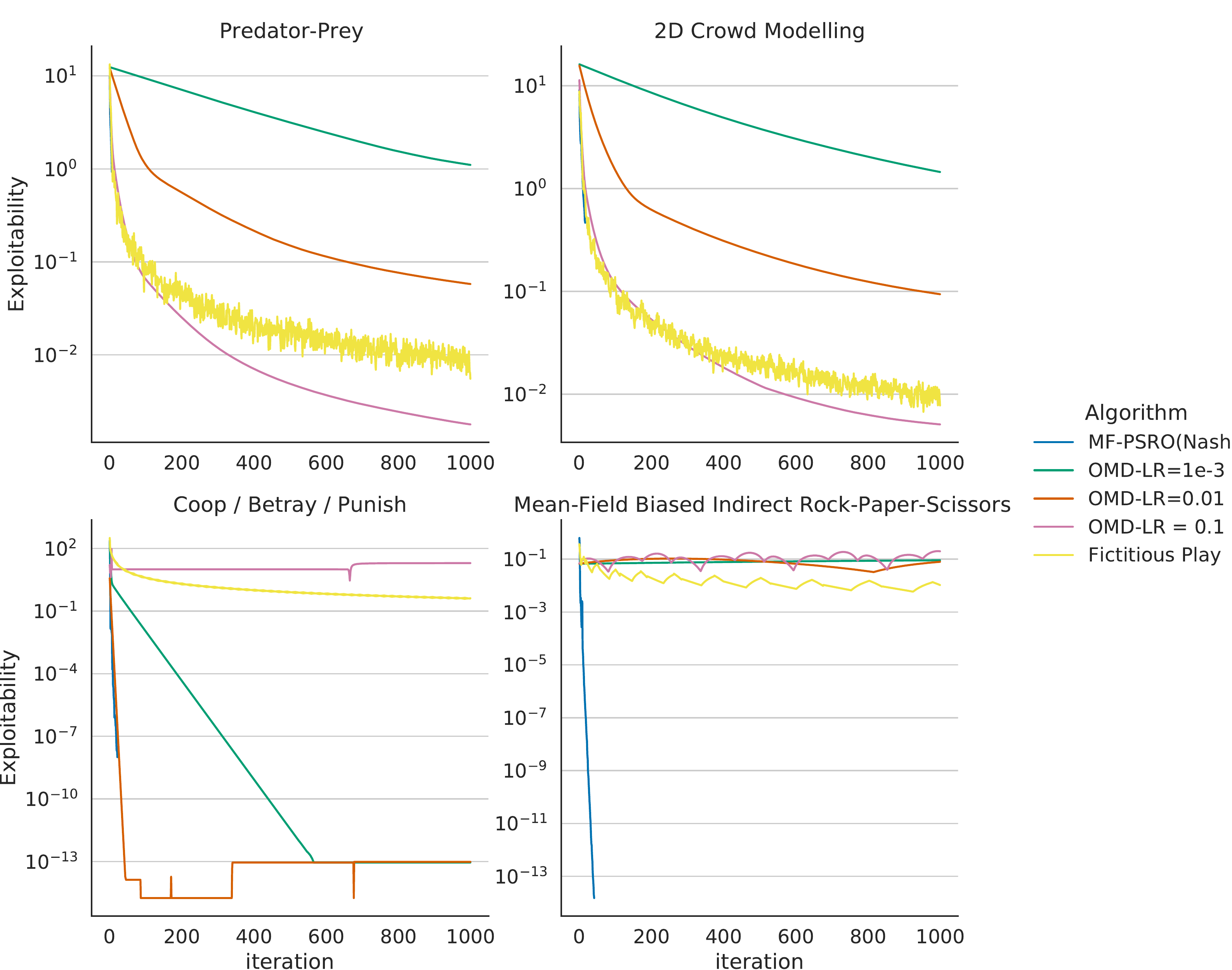}
    \caption{Algorithm Exploitabilities}
    \label{fig:exploitabilities_zoomed_out}
\end{figure}
\end{centering}

\section{On No-Regret Learners, Bandit Compression and Sped-Up mean-field PSRO}\label{appendix:no_regret_bandit_compression_and_psro}

\subsection{Commentary on Complexity Analysis}\label{subappendix:complexity_comments}

\begin{itemize}
    \item Minimizing swap regret directly has higher complexity than payoff matrix estimation by a factor $n^2 log(n)$ in worst cases.
    \item Minimizing external regret directly has higher complexity than payoff matrix estimation by a factor $log(n)$ in worst cases.
    \item Using regret minimizers directly provides the user with a useable distribution over policies
    \item Estimating the payoff matrix means the user still has to run an algorithm over said payoff matrix, which could have large complexity (Linear solvers have complexity $O(n^{2 + \gamma})$ with $\gamma > 0$, for example)
    \item The relationship between payoff uncertainty and solver output uncertainty is difficult to analyze in general, due to the strong non-linearities of solvers. Indeed, picture the following 0-sum game: three strategies face off, $\pi_1$ has payoff 100 against $\pi_2$ and $\pi_3$, and $\pi_2$ has payoff 1 against $\pi_3$. Any reasonable $\epsilon$-error in estimating the payoff obtained by playing $\pi_1$ would not change e.g. its Nash distribution, or the distribution of a correlated equilibrium. In contrary, in a game where all average payoffs are very close to 0 ($\pi_1$ barely beats $\pi_2$ and $\pi_3$, and $\pi_2$ barely beats $\pi_3$), an $\epsilon$-error could lead to a reversal of these interactions (In the noisy payoff matrix, it could be that $\pi_2$ beats $\pi_3$ which beats $\pi_1$), thus completely changing the computed distribution.
    \item We do not yet fully understand the complexity reduction granted by Bandit Compression, which could greatly lower asymptotic complexity of the Bandit approach.
    \item This complexity insight can be transferred to the N-player case. In this case, one needs to compute $(n+1)^N - n^N = O(n^{N-1})$ matches, and the estimation complexity is therefore $T = O\left(\frac{n^{N-1}}{\alpha \epsilon^2}\right)$. The number of different actions is $n^N$, so the complexity of minimizing internal regret is $O\left(\frac{N n^{3N} log(n)}{\epsilon^2}\right)$ and it is $O \left (\frac{N n^N log(n)}{\epsilon^2} \right)$ for external regret minimization.
    \item If we can observe all policies' payoffs at no additional cost, the regret bounds become $O\left(\frac{n \; log(n)}{\epsilon^2}\right)$ for internal regret, and $O\left(\frac{log(n)}{\epsilon^2}\right)$ for external regret. The complexity for payoff matrix estimation nevertheless remains $O(\frac{n+1}{\alpha \epsilon^2})$, as one is interested in $(J(\pi_n, \mu^{\pi_k}))_k$ and $(J(\pi_k, \mu^{\pi_n}))_k$
\end{itemize}

\subsection{Improvements on sped-up PSRO algorithm}\label{subappendix:speed_improvement_details}

Sped-Up PSRO includes the following new features:

\begin{itemize}
    \item $\rho_{tol}$: This term is a regret threshold. If the optimal solution of Problem \ref{eq:bandit_cce_optim_problem} or \ref{eq:bandit_ce_optim_problem} yields regret lower than this term, we consider the equilibrium search as successful.
    \item $\rho_{lim}$ and new loop conditions: The PSRO loop does not terminate anymore when $\Pi_{n+1} == \Pi_{n}$, what it does then is \emph{refine} its current equilibrium by halving $\rho_{tol}$ at every iteration where $\Pi_{n+1} == \Pi_{n}$, until $\rho_{tol} == \rho_{lim}$, with $\rho_{lim}$ set to a very low value, typically $10^{-12}$
    \item $\tau_{Compress}$: Typically set to 1, this value allows one not to optimize problems \ref{eq:bandit_cce_optim_problem} or \ref{eq:bandit_ce_optim_problem} at every regret minimization step. This can be especially useful when computing MFCEs, for which the problem is much slower to solve than for MFCCEs.
\end{itemize}

\subsection{Proof of Complexity in Noisy Payoff case}\label{subappendix:proof_noisy_payoff}

In the case where payoffs are additively noisy, Regret can be decomposed into two terms: $\text{Regret}_i = R_i + \tilde R_i$, where $R$ is the true, noiseless regret, and $\tilde R_i$ is the noise-derived regret. Given $\alpha > 0$, after $O\left(\frac{M \; n^{3} \; log(n)}{\alpha^2}\right)$ steps, we know that $\text{Regret} \leq \frac{\alpha}{2}$.
We have that $\text{True Regret} - \text{Regret} = \max_j R_j - (\max_i R_i + \tilde R_i) \leq \max_i -\tilde R_i$, and $\mathbb{P}(\max_i -\tilde R_i \geq \alpha) \leq \sum_i \mathbb{P}(-\tilde R_i \geq \alpha)$. Given that $\tilde R_i = \frac{1}{T} \sum_t \epsilon_i[t] - \nu_t^t \epsilon[t]$, and $\epsilon_t$ is averaged over $M$ samples and thus has $\frac{\sigma^2}{M}$ variance, then $\text{Var}(-\tilde R_i) \leq \frac{\sigma^2}{TM}$ and Chebyshev's inequality yields $\mathbb{P}(-\tilde R_i \geq \frac{\alpha}{2}) \leq \frac{4\sigma^2}{TM \alpha^2}$, thus yielding $\mathbb{P}(\max_i -\tilde R_i \geq \alpha) \leq n \frac{4\sigma^2}{TM \alpha^2}$.
We then have $\mathbb{P}(\text{True Regret} \leq \alpha) \geq \mathbb{P}(\text{Regret} + \max_i -\tilde R_i \leq \alpha) \leq \mathbb{P}(\max_i -\tilde R_i \leq \frac{\alpha}{2}) \leq 1 - n \frac{4\sigma^2}{TM \alpha^2}$ The probability of the true regret being lower than $\alpha$ after $O\left(\frac{M \; n^{3} \; log(n)}{\alpha^2}\right)$ steps is therefore $\delta \geq 1 - n \frac{\sigma^2}{TM \alpha^2}$. Since each regret steps is now composed of $M$ times as many rollouts, the total rollout-complexity of the algorithm must be multiplied by $M$, which concludes the proof.

\subsection{On the complexity of improving Welfare}\label{subappendix:welfare_complexity}

We have so far introduced a method that learns Nash, correlated and coarse correlated equilibria in mean-field games. A subsequent question for correlated and coarse correlated equilibria is, could we influence the learning process for it to find high-welfare equilibria instead of low-welfare ones ?

This problem of high-welfare convergence was shown by \cite{barman2015finding} to be NP-hard in general, with the notable exception of succinct aggregate games, for which the existence of polynomial algorithms converging to high-welfare equilibria is proven. However, their method relies on a discretization of and grid-search over the aggregate space, the space of statistics summarizing the behavior of other players. 

At the $n$-th step of mean-field PSRO, discretizing the $n$-dimensional probability vector space with step size $\frac{1}{M}$ amounts to considering matrices of size $ \geq n \left(\frac{M}{2}\right)^n$, a complexity exponential in the number of iterations, therefore prohibitive.

We therefore leave open the question of high-welfare convergence for now.

\subsection{Optimality of $\rho_*$ for correlated equilibria}\label{subappendix:rho_star_optimality_ce}

For restricted MFCEs:

The deviation incentive against policy $\pi_i$ recommended by the correlation device sampling $\nu_t$ with probability $\rho_t$ in the restricted game is 
\begin{align*}
    \max_{\pi, \pi'} \rho(\pi) \mathbb{E}_{\nu \sim \rho(\cdot | \pi)} &\left[ J(\pi', \mu(\nu)) - J(\pi, \mu(\nu)) \right] \\
    &= \max\limits_{i, j} \sum\limits_t \rho_t \nu_t(i) \Big( J(\pi_j, \mu(\nu_t)) - J(\pi_i, \mu(\nu_t)) \Big) \\
    &= \max\limits_{i, j} \rho^t \text{Regret}_{i, j}
\end{align*}
Since the uniform distribution is a possible solution of Problem \ref{eq:bandit_ce_optim_problem}, we thus have that the average max deviation incentive against $\rho^*$ the solution of Problem \ref{eq:bandit_ce_optim_problem} is lower than or equal to that of the uniform distribution, which concludes this part of the proof.

\section{On Monotonicity in Restricted Games}

\begin{property}
    A game is monotonic if and only if all its restricted games are.
\end{property}
\begin{proof}
    Assume all restricted games are monotonic, take $\pi_1, \pi_2$ two policies of the game, and take the monotonic game containing only $\pi_1, \pi_2$. By assumption, it is monotonic, ie. $\forall \nu_1, \; \nu_2 \in \Delta\left(\{ \pi_1, \pi_2 \}\right)$, $$J_r(\nu_1, \nu_1) + J_r(\nu_2, \nu_2) - J_r(\nu_1, \nu_2) - J_r(\nu_2, \nu_1) \leq 0$$ with $J_r(\nu, \nu') = J(\pi(\nu), \mu(\nu'))$. It suffices to take $\nu_1 = \delta_{\pi_1}$ and $\nu_2 = \delta_{\pi_2}$ to directly have $$J(\pi_1, \mu^{\pi_1}) + J(\pi_2, \mu^{\pi_2}) - J(\pi_2, \mu^{\pi_1}) - J(\pi_1, \mu^{\pi_2}) \leq 0$$ and since this is true for all $\pi_1, \; \pi_2$, the game is monotonic.
    
    Assume the game is monotonic. Take $\pi_1, ..., \pi_N$ with $N > 0$ and consider their derived restricted game. Let $\nu_1, \nu_2 \in \Delta(\{ \pi_1, ..., \pi_N \})$. $$J_r(\nu_1, \nu_1) + J_r(\nu_2, \nu_2) - J_r(\nu_2, \nu_1) - J_r(\nu_1, \nu_2) = J(\pi(\nu_1), \mu(\nu_1)) + J(\pi(\nu_2), \mu(\nu_2)) - J(\pi(\nu_2), \mu(\nu_1)) - J(\pi(\nu_1), \mu(\nu_2))$$ given that $\forall \nu, \; \mu(\nu) = \mu^{\pi(\nu)}$. Since $\pi(\nu_1)$ and $\pi(\nu_2)$ are both policies of the true game, and the true game is monotonic, $$ J(\pi(\nu_1), \mu(\nu_1)) + J(\pi(\nu_2), \mu(\nu_2)) - J(\pi(\nu_2), \mu(\nu_1)) - J(\pi(\nu_1), \mu(\nu_2)) \leq 0$$ and thus $$ J_r(\nu_1, \nu_1) + J_r(\nu_2, \nu_2) - J_r(\nu_2, \nu_1) - J_r(\nu_1, \nu_2) \leq 0 $$ which concludes the proof.
\end{proof}

\end{document}